\documentclass[twocolumn, journal]{IEEEtran}
\IEEEoverridecommandlockouts                             
\usepackage{setspace}      
\usepackage{amsmath}   
\usepackage{amssymb}     
\usepackage[dvips]{graphicx}  
\usepackage{color}
\usepackage[tight]{subfigure}
\usepackage{bm}   
\usepackage{booktabs}  
\usepackage[boxruled, linesnumbered]{algorithm2e}
\usepackage{url}
\usepackage{longtable}
   
\usepackage{cuted}
   
\usepackage{mathtools}     
\usepackage{cancel}  
  
\usepackage{soul}  
\usepackage[normalem]{ulem}

\newtheorem{RK}{Remark}   
\usepackage{algorithm2e}
  
\newtheorem{define}{Definition}
\newtheorem{prop}{Proposition}
\newtheorem{lem}{Lemma}
  
\newtheorem{thm}{Theorem}

\newtheorem{cor}{Corollary}

\newcommand{\mc}{\mathcal}
\newcommand{\mb}{\mathbf}

\newcommand{\mbb}{\mathbb}

\title{Distribution System Topology Detection Using Consumer Load and Line Flow Measurements}
\author{Raffi Avo Sevlian and Ram Rajagopal \thanks{R. Sevlian is with the Department of Electrical Engineering and the Stanford Sustainable Systems Lab, Department of Civil and Environmental Engineering, Stanford University, CA, 94305. Email: rsevlian@stanford.edu.}
\thanks{R. Rajagopal is with the Stanford Sustainable Systems Lab, Department of Civil and Environmental Engineering, Stanford University, CA, 94305. R. Rajagopal is supported by the Powell Foundation Fellowship. Email: ramr@stanford.edu.} }
\begin{document}
\maketitle

\begin{abstract}
This work presents a topology detection method combining home smart meter information and sparse line flow measurements.
The problem is formulated as a spanning tree detection problem over a graph given partial nodal and edge flow information in a deterministic and stochastic setting.
In the deterministic case of known nodal power consumption and edge flows we provide sensor placement criterion which guarantees correct identification of all spanning trees.
We then present a detection method which is polynomial in complexity to the size of the graph.
In the stochastic case where loads are given by forecasts derived from delayed smart meter data, we provide a combinatorial Maximum a Posteriori (MAP) detector and a polynomial complexity approximate MAP detector which is shown to work near optimum in low noise regime numerical cases and moderately well in higher noise regime.
\end{abstract}
  IEEE Transactions on Control of Network Systems
\section{Introduction}  

The need for advanced controls in the distribution system is an emerging topic in power system and  controls communities.
Proposed computational models for problems such as dispatching of distributed energy resources \cite{Lavaei2012}, \cite{Lam2011} or coordinated voltage control \cite{Farivar2012}, \cite{Lam2012}, \cite{Jahangiri2013}, \cite{Smith2011} assume known system topology and network parameters.
In reality customer level feeders are not known with an accuracy equivalent to that of the transmission system.

Enabling improved management and control requires significantly improved estimation of the system state.
This can be illustrated in the IEEE 123 Bus System shown in Figure \ref{fig:ieee123-feeder}.
The network not only has end nodes which represent residential transformers (blue rectangle), but various switching devices (green rectangles) and four feeders (red circles).
In this system, estimating the system state requires the determination of the voltage phasor at every node and the status of all discrete devices that can connect and disconnect loads.
In such a setting a Generalized State Estimator (GSE) \cite{Monticelli2000} is used to determine the $\{0, 1\}$ of each discrete device as well as the voltage at each bus.

Some previous work has presented solutions to this issue, which differ from the contributions of this work.
In  \cite{Korres2012}, a traditional GSE is employed to identify the correct topology in a distribution system.
The work presents a traditional weighted least square state estimator and use dummy variables for breaker status indicators, and assume knowledge of the system line parameters.
The measurement types are focused on substation SCADA and load measurements in a very simple network. 
In \cite{Arghandeh2015} the authors introduce the use of high frequency micro-Phasor Measurement Unit ($\mu$PMU) data in the topology detection task.
They propose a method of comparing simulation and measured $\mu$PMU data for each topology.
This work assumes high frequency voltage magnitude and phase measurements are available in each bus.
In \cite{Cavraro2015, Deka2015A} and \cite{Deka2015B} the authors develop a voltage time series approach to identifying topology changes relying on voltage data at each home.
However they rely on long time captures, so their method is more in line of network discovery not real time topology detection.
In \cite{Sharon2012} the authors present a general state estimator based method that is used in topology detection, similar to \cite{Korres2012}.

\begin{figure}[h]
\includegraphics[scale=0.5]{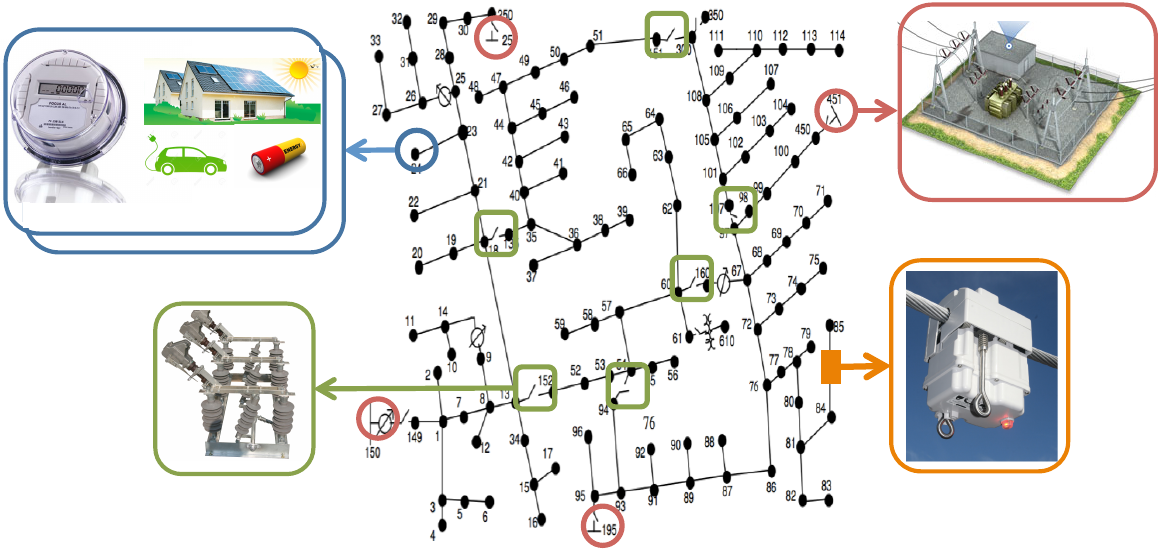}   
\label{fig:ieee123-feeder}
\caption[A set of four subfigures.]{
IEEE 123 Distribution system with commonly occurring sensing and actuating technologies, illustrated in clockwise order: Switching device, smart meters, substations and line sensors.
}
\end{figure}
   
The contributions of this work differs significantly from prior work in the following ways.
Fist, we assume the following information is available: 1) widely available load measurements from smart meters; 2) line flows on a fraction of the lines obtained from either line sensing or substation SCADA.
The line measurements are typically available in real time, while smart metering data is delayed by multiple hours requiring some forecasting if real time topology detection is required.
Second, the detection and sensor placement problems are developed in both the deterministic case of combining historical load and line data and the stochastic case of combining line measurements with load forecasts.
Additionally, our model assumes a lossless network, which although introducing some error is much smaller than the typical load forecast.
We show this solution lends itself well to a very robust data driven approach where many of the line parameters are not known, or when AMI connectivity information may be in error.
This robustness under large uncertainty makes this a very practical and useful method for utilities.

The paper is organized as follows.
Section \ref{section-Problem-Formulation}, \ref{section-System-Model} formulates the problem of topology detection.
Sections \ref{section-deterministic-case} and \ref{section-stochastic-case} solve the detection and sensor placement problems in the deterministic and stochastic cases respectively.
Finally, numerical demonstrations are given in Section \ref{section-Numerical-Experiment}, with additional details in the Appendix.
  
\section{Problem Formulation} 
\label{section-Problem-Formulation}

Consider a power distribution network where multiple feeders can supply energy to all consumers, and must be operated in a radial structure at all times.
In network reconfiguration, sets of breakers and tie switches can reconfigure themselves such that all loads are connected and no feeders are connected.
The task of recovering the network topology is to detect the switch statuses given all available information.

\subsection{DC Power Flow}
\label{subsection-distribution-system-model}

We use a DC power flow approximation to the actual AC flow in the distribution system {\color{black} \cite{Stott2009dc}}.
The model is normally used in approximating the voltage magnitude and phase in the network, but since our detection problem relies on power flows, this is equivalent to using a lossless \textit{network flow} representation.

In a usual representation, the distribution system is modeled as a graph $G(V, E)$ where vertices, $v \in V$ represent nodes (transformers) and the edges $e \in E$ represent the distribution lines.
The signed incidence matrix is $B \in \{-1, 0, +1\}^{|V| \times |E|}$, where each undirected edge has a pre specified direction: $e_k = (v_n, v_m)$ on which to assign columns of $B$ as follows:
\begin{align}
 B_{i, j} = \begin{cases}
   	 	+1 &~\text{if  $v_i$ is the \textit{originating node} of edge $e_j$}  \\
                 -1 &~\text{if $v_i$ is the \textit{terminal node}  of edge $e_j$}  \\ 
                  0 &~\text{else}. \end{cases}
\end{align}
Given the set of net injections in the network, $\mb{y}$ the flow constraints can be represented as: $B \mb{f} =  \mb{y}$.
This can be extended to a complex load case but is out of the scope of this work.
\subsection{Load Model}
\label{subsection-measurement-model}
Each load $v_{n}$ in the system has a consumption $x_{n}$. 
We assume that the loads are single phase real power quantities and the forecast errors are independent random variables: $\epsilon_n\sim N(0, \sigma_{n}^2)$ and $x_{n} \sim N(\hat{x}_{n}, \sigma^{2}_{n})$.   
Given the single global source of energy, we have the following $\mb{y} = [\mb{1}^T\mb{x}~-\mb{x} ]^{T}$.%
\subsection{Switching Model Network Configuration}
\label{subsection-switching-model-network-configuration}
Each switch has a status $w_i \in \{0, 1\}$, and $\mb{w} = \{w_1, \hdots, w_{K} \}$.
The switching is constrained so that all loads must be connected to some feeder and there can exist no path between various feeders.
This ensures that each feeder is connected to some set of loads in a radial configuration, and that no loads are in outage.  
\subsection{Measurement Model}
\label{subsection-Measurement Model}
For any edge $e$ of the original distribution system, we denote by $s$ the power flow on it to all active downstream loads. 
The sensor placement $\mc{M} \subset E$, is a subset of edges of the network.
We assume that the magnitude and direction of power flow is measured.
Additionally, we assume that the power flow measurements are error free.
This assumption can be made since any instrumentation error will be much smaller than the pseudo-measurement errors in practice.

Given a topology defined by $\mb{w}$, the set of all measurements is $\mb{s}$ where the $k^{th}$ is given by
\begin{align}
s_k(\mb{w},~\mb{x}) = \sum_{j: v_j \in V_{k}(\mb{w}) } x_{j}.
\end{align}
The set $V_{k}(\mb{w})$ is the subset of nodes for a particular topology downstream of $k^{th}$ flow measurement under switch state \textbf{w}.

\subsection{Topology Detection}
\label{subsection-Topology Detection}
 
The detection and placement problem is solved in two scenarios: (1) \textit{deterministic case}, where loads and flows are known perfectly, (2) \textit{stochastic case}, where loads are known with uncertainty due to forecasting error.

In the deterministic case, a simple detector will return all topologies which satisfy the load and flow information as follows:
\begin{align}
\hat{\mb{w}} = \{ \mb{w} \in \{0, 1\}^{K} : \mb{s}_{obs} = \mb{s}(\mb{w}, \mb{x}) \}. \label{opt1-combinatorial-detector}
\end{align}
In the stochastic case, a MAP detector can be written as 
\begin{align}
\hat{\mb{w}}  \in \underset{w_i \in \{0, 1 \}^{K}  }{\arg\max} \Pr\left( \mb{w}~|~\mb{s},~\mb{\hat{x}}  \right). \label{eq-switch-MAP-detector} 
\end{align}
  
These naive methods are inefficient and provides no guarantee on unique detection or sensor placement.
For both detector types the following general questions are explored.
\begin{enumerate}
\item \textit{(Correctness)}: How to guarantee that this method will return a unique and correct spanning tree?
\item \textit{(Efficiency)}:     How to search for the correct configuration without evaluating all $2^{K}$ configuration since this can be inefficient?
\item \textit{(Sensor placement)}: Where to place line sensing to minimize missed detections in both deterministic and stochastic settings?
\end{enumerate}

The following sections show how this problem can be reduced to a spanning tree detection problem, and how it can be solved in an efficient matter and provide some guarantees on sensor placement for correct status recovery.

\section{Model Reduction}
\label{section-System-Model}
We show that the general distribution system with switching devices under a lossless power flow can be reduced to an island graph which simplifies the structure of the valid configurations.
The detection problem is then cast as a spanning tree detection problem with nodal and edge measurements on the island graph.
\subsection{Island Graph}
\label{subsection-island-graph}

\begin{figure}[h]
\hspace{-5mm}
\subfigure[][]{ 
\includegraphics[scale=0.37]{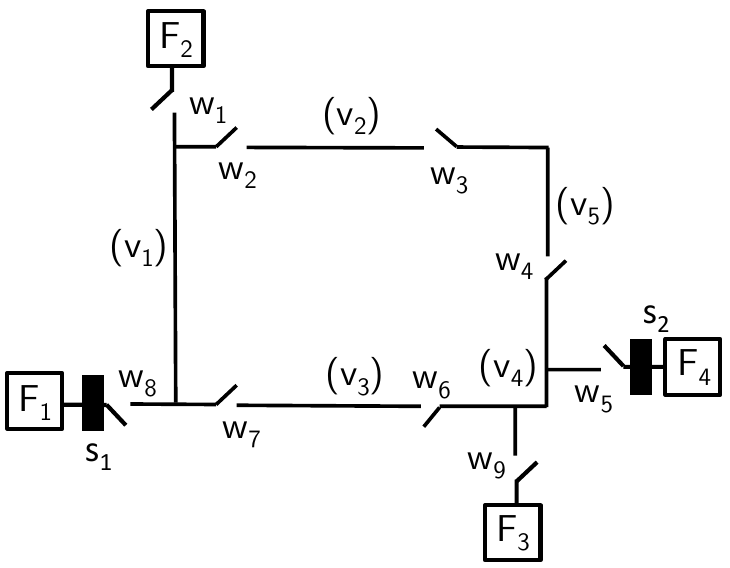}   
\label{fig:reduced_feeder}  
}
\subfigure[][]{ 
\includegraphics[scale=0.37]{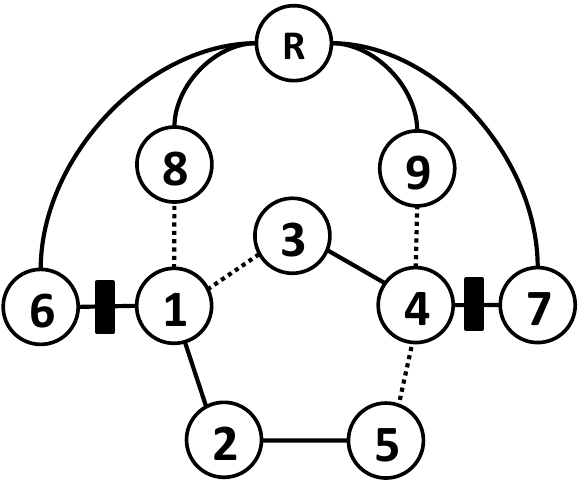}  
\label{fig:island_graph_with_spanning_tree}
}
\caption[  Reduced Feeder and Island Graph  ]{
\subref{fig:reduced_feeder} Typical test feeder with sectionalizing switch operation.
\subref{fig:island_graph_with_spanning_tree} Island graph simplifying topology of feeder.
}
\end{figure}
Consider again the IEEE 123 node feeder in Figure \ref{fig:ieee123-feeder}.
The loads which are connected to each other and separated by switches can be reduced to a set of connected islands separated by various switches.
This reduced representation is shown in Figure \ref{fig:reduced_feeder} where all connected regions are grouped into single lines for visual simplicity.  

For example, sources $150$, $251$, $195$ and $451$ in Figure \ref{fig:ieee123-feeder} are feeders $F_1$, $F_2$, $F_3$, $F_4$ in Figure \ref{fig:reduced_feeder}.
The switch constraints can be seen easily here. 
For example, both $w_1$ and $w_8$ being closed will violate the radial structure of the network since the two feeders will be energizing the same set of loads.
\begin{table}[h]
\centering
\caption{Mapping of vertices and edges to construct island graph.} 
\label{tab:model-comparison}
\begin{tabular}{@{}cccl@{}}
\toprule
\hspace{1mm} Island Graph & & &  IEEE Test Feeder    \\
\cmidrule{1-2}                                     \cmidrule{4-4}
Switch   &                 Node 		& &   Load ID         \\
$w_1$  &   $F_2 ,v_8 - v_1$         	& & $(250 - 251)$  \\ 
$w_2$  &   $v_1 - v_2$         		& & $(18-135)$      \\ 
$w_3$  &   $v_2 - v_3$         	 	& & $(151- 300)$   \\ 
$w_4$  &   $v_4 - v_5$          		& & $(97 -197)$     \\ 
$w_5$  &   $v_4 - F_3, v_7$          	& & $(450 -451)$   \\ 
$w_6$  &   $v_3 - v_4$          		& & $(54 - 94)$      \\ 
$w_7$  &   $v_1 - v_3$          		& & $(13 - 152)$    \\ 
$w_8$  &   $F_1, v_6 - v_1$          	& & $(149 - 150)$  \\ 
$w_9$  &   $F_1, v_6 - F_4, v_9$   & & $(95 - 195)$       \\ 
\bottomrule
\end{tabular}
\label{tab:feeder_to_reduced_island_graph}
\end{table}
  
This representation can be further simplified to an \textit{island graph} using the following steps:
\begin{enumerate}
\item loads from the island graph are converted to vertices in the graph;
\item feeders $F_1 \hdots F_4 $ are turned into vertices;
\item switches in {\color{black}the Island} Graph are converted to undirected edge in the graph;
\item a virtual root node and directed edges $(v_r \rightarrow F_j)$ for all feeders are added.
\end{enumerate}
The result is the island graph in Figure \ref{fig:island_graph_with_spanning_tree}.
A complete mapping between the reduced feeder and the island graph is given in Table \ref{tab:feeder_to_reduced_island_graph}.
The island graph $G = (V, E)$ is the network used in the remaining analysis.   
We denote the added edges $\tau = \{ e \in E: e=(v_r, F_i)~\forall F_i \}$. 
This construction leads to a simple method for enumerating each valid topology of the island graph.

\subsection{Switch Configurations via Island Graph}
\label{subsection-swtich-configurations-via-island-graph}

Consider $\mc{T}$ to be any spanning tree over $G$ and $\mbb{T}$ the set of all spanning trees constructed on $G$.
We refer to the set of spanning trees containing subtree $\tau$ as $\mbb{T}_{\tau}$.
Figure \ref{fig:island_graph_with_spanning_tree} represents an example spanning tree that can be constructed in the island graph $G$.
The following relationship makes our representation useful in the detection task.
\begin{prop}
\label{prop-spanning-tree-switch-config}
The set $\mbb{T}_{\tau}$ represents all valid switch configurations $\mb{w}$ in the reduced network.
\end{prop}   

Each switch status can be mapped to some spanning tree in $\mbb{T}_{\tau}$.
We will use $\mc{T}(\mb{w})$ as shorthand for the tree corresponding to $\mb{w}$.
This representation is now used to develop the topology detection problems in both deterministic and stochastic settings.

\subsection{Load and Line measurements in the Island Graph}
\label{subsection-load-and-line-measurements-in-the-island-graph}

The line and  load measurements are analyzed in the island graph as follows:
\begin{itemize}
\item [] \textit{Load Measurements}, where a network flow model is used with no losses, the total consumption in an island is the sum of all nodes in the original graph;
\item [] \textit{Line Measurements} map to edge flow measurements on the island graph.  
\end{itemize}
The second fact is due to the following. 
Measurements can occur in the middle of an island, or at a switch location.
If the measurement is taken at a switch, it corresponds directly with an edge measurements as claimed.
If the measurement is taken inside a load island, we can create a virtual edge in the island graph and add it to $\tau$ so as to restrict $\mbb{T}_{\tau}$.
The results do not change in the stochastic or deterministic cases, but complicate the analysis.
For simplicity, we will only assume line flows are monitored at switches.


\section{Deterministic Case}  
\label{section-deterministic-case}

\subsection{Deterministic Detector}
\label{subsection-deterministic-detector}
Given nodal consumptions $\mb{x}$ and observed flow $\mb{s}_{obs}$, the following program can be used to solve the deterministic detector problem using the island graph formulation:
\begin{align}	
& \text{find}~ \mb{f}, \mb{w}				\label{opt1} \tag{OPT-1}  		\\
& \text{s.t.} 							\nonumber 				\\
& ~~~~~| f_i | \leq |\mb{x}| w_i  				\label{opt1-bound-flows}   		\\ 
& ~~~~~B \mb{f} = \mb{y}        				\label{opt1-network-flow}   	\\ 
& ~~~~~ A_{\mc{M}} \mb{f} = \mb{s}_{obs}        \label{opt1-flow-obs}  		\\
& ~~~~~ \mc{T}(\mb{w}) \in \mbb{T}_{\tau}.       \label{opt1-sp-tree}  		   
\end{align}

\eqref{opt1} is a mixed integer program with boolean $\mb{w} \in \{0, 1\}^{K}$ for the edges in $G$ (switch statuses) and the scalar $\mb{f}$ flow along each edge feeding the loads.
Eq. \eqref{opt1-bound-flows} limits the edge flows to either be set to zero, or be fully unconstrained according to the topology.
Eq. \eqref{opt1-network-flow} is the network flow constraint relating flows to load measurements.
Eq. \eqref{opt1-flow-obs} sets each observed edge to the sensor value while  \eqref{opt1-sp-tree} constrains the status of edges in $G$ to form a spanning tree.
Matrix $A_{\mc{M}}$ indicates the edges that are being measured, so $A_{\mc{M}}(k, m_k) = 1, \forall e_{m_k} \in \mc{M}$.

A naive solution to \eqref{opt1} will enumerate every spanning tree, then given the nodal consumptions, evaluate the theoretical flow value $\mb{s}(\mc{T}, \mb{x})$ and compare it to the observed flow.
At this point, the algorithm complexity is reduced from $2^{K}$ to $O(|\mbb{T}_{\tau}|)$.
\subsection{Spanning Tree Identifiability}
\label{subsection-spanning-tree-identifiability}

The following section provides the conditions in which a naive detection procedure can recover the correct and unique solution.
This corresponds to a line sensor placement which guarantees a unique solution.
First, the following definition is of use.
\begin{define}
The set of spanning trees, $\mbb{T}$, is identifiable if $\forall \mc{T}, \mc{T}^{\prime} \in \mbb{T}$ where $\mc{T} \neq \mc{T}^{\prime}$ we have that $\bf{s}(\mc{T}, \mb{x}) \neq \bf{s}(\mc{T}^{\prime}, \mb{x})$. \label{def-tree-identifiable}
\end{define}

In the deterministic case, we desire a placement $\mc{M}$ such that $\mbb{T}$ is identifiable.
This is referred to as a \textit{valid placement}.
This serves as a baseline to investigate the stochastic case and provides intuition for the problem.
A naive method of evaluating whether the placement is valid is to evaluate $\mbb{T}$, then test whether any two trees in the set evaluate to the same observation.
This naive procedure has $O(|\mbb{T}|^2)$ complexity and provides no insight.
The following theorem provides the necessary and sufficient conditions in which $\mbb{T}$ is identifiable and a placement is valid.

\begin{thm}   
\label{thm:spanning_tree_identifiability}
$\mbb{T}$ is identifiable if and only if the graph $G \setminus \mc{M}$ of the island graph forms a spanning tree.
\end{thm}

The intuition of Theorem \ref{thm:spanning_tree_identifiability} is that to have observability of all spanning trees, we must have a sensor placement such that any cycle that can be constructed on $G$ will have some flow sensor on it. 
This corresponds to the dimension of the cycle space, referred to as the circuit rank $\mu = |E| - |V| + 1$, which is the minimum number of measurements needed to correctly detect all spanning trees on $G$.
This gives us a $O(E)$ verifiable condition to ensure that all spanning trees are identifiable as opposed to $O(|\mbb{T}|^2)$ with the naive method.

Theorem \ref{thm:spanning_tree_identifiability}  provides a way to construct the set of all placements where identifiability is achieved.
First consider the function $h(\mc{T}) = E \setminus \mc{T}$, which returns the edges in $G$ that are not in $\mc{T}$.
These edges are referred to as the co-tree of $\mc{T}$.
An obvious consequence to Theorem \ref{thm:spanning_tree_identifiability} is the following:
\begin{cor}
The function $h: \mc{T} \rightarrow \mc{M}$ is a bijection between the set $\mbb{T}$ and the set of all valid placements $\mbb{M}$.
\label{thm-tree-placement-mapping}
\end{cor}  
The following is useful:   
\begin{RK}
Corollary \ref{thm-tree-placement-mapping} implies that $|\mbb{M}| = |\mbb{T}|$.
\end{RK}

Corollary \ref{thm-tree-placement-mapping} is quite important from a placement perspective since it actually yields a method to generate a valid placement in the deterministic case.
Also, it allows us to enumerate \textit{all} valid placements for a graph.  
This is important when dealing with a stochastic case where sensor placement relies on mostly evaluating each placement in $\mbb{M}$.
In the case of valid placements on the island graph this set is restricted, since having sensors on edges in $\tau$, {\color{black} would have no physical meaning.}
The restricted set is given by 
\begin{align}
\mbb{M}_{\tau} = \{ G \setminus \mc{T} | \mc{T} \in \mbb{T}_{\tau} \} \label{eq:generate-valid-placement-restricted}.
\end{align}

\subsubsection{Spanning Tree Detection Without Flow Direction}
\label{subsubsection-spanning-tree-detection-without-flow-direction}. 

The sensor placement condition in Theorem \ref{thm:spanning_tree_identifiability} assumes that line flow magnitude and direction are known.
This may not be the case in some line sensing situations where only the magnitude is known but not the direction, since the phase difference between voltage and current must be known for this.

The placement $\mc{M} = \mc{M}_1 \cup \mc{M}_2$ is such that $\mc{M}_1$ satisfies the condition in Theorem \ref{thm:spanning_tree_identifiability} and $|\mc{M}_2| \geq 0$.
For the added measurements, we develop sufficient conditions on $\mc{M}_2$ so that $\mbb{T}(G)$ is identifiable.
First, consider the spanning tree $\mc{T} = G \setminus \mc{M}_1$, and the fundamental cycle basis (See Appendix \ref{section-Useful-Graph-Theory-Definitions-and-Results}), $\mc{FC}_{\mc{M}}$. 
Next consider all the discrete path of edges formed in cycle $c_i \in \mc{FC}_{\mc{M}_1}$, which do not belong to any other cycle given as $p_i = c_i \setminus \cup_{j \neq i} c_j$.
The following sufficient condition on added measurements $\mc{M}_1$ leads to spanning tree identifiability.

\begin{thm}     
\label{thm:undirected-flow-placement}
For any $|p_i| \geq 3$, where $m_i \in \mc{M}_1$ is not on an endpoint, an additional measurement $m_i$ is required on some edge in $p_i$.
\end{thm}  

This condition implies that in the worst case, $2 \mu$ flow sensors are required to uniquely distinguish any potential spanning tree.
Since this is a sufficient condition, there can exist many placements $\mc{M}_1$, where many of the $p_i$'s are of length $1$ or $2$.
Therefore, deterministic placement can be performed by computing all $\mc{FC}_{\mc{M}_1}$ and find the placement with smallest number of $|p_i|>2$.

\subsection{Spanning Tree Detection via Relaxed Flow Solution}   
\label{subsection-deterministic-and-placement}

Theorem \ref{thm:spanning_tree_identifiability} provides a condition where a unique solution to \eqref{opt1} can be found but provides no efficient method to find it beyond exhaustive search.

It can be shown that solving a relaxed form of \eqref{opt1}, without {\color{black} Boolean constraints can recover the correct topology:}
\begin{align}
\mb{f}^{\star} = \{\mb{f}:~\text{st.}~B~\mb{f} = \mb{y}~\text{and}~f_{i} = s_{obs, i}~\forall e_i \in \mc{M} \} \label{relaxed-flow-solution}.
\end{align}

The solution to the linear equation over $\mb{f} \in R^{|E|}$ recovers the sparsity pattern in $\mb{f}$ corresponding to a spanning tree without any {\color{black} sparsity-inducing heuristics.}

We can represent the network flow by partitioning the incidence matrix $B$ and flow vector $f$ into observed $(B^r_M, \mb{f}_M)$ and non-observed $(B^r_N, \mb{f}_N)$ components.
Where $B^r_{N}$ and $B^r_{M}$ are the matrices with their first row removed.
This results in the following:
\begin{align}
 \left[ \begin{array}{cc} B^r_{N} & B^r_{M} \\ 0 & I \\  \end{array} \right] \left[ \begin{array}{c} \mb{f}_N \\ \mb{f}_M \\ \end{array}  \right]  = \left[ \begin{array}{c} \mb{x} \\ \mb{s}_{obs} \\ \end{array}  \right].  \label{network-flow-constraint-matrix} 
\end{align}
\begin{lem}
\label{lem-BN-rank-N-1}
For the sensor placement condition in Theorem \ref{thm:spanning_tree_identifiability}, the matrix $B^r_N$ is has $\text{rank}(B^r_N) = N-1$, and is invertible. 
\end{lem}
From Lemma \ref{lem-BN-rank-N-1}, the following can be computed:
\begin{align}  
\mb{f}^{\star}(\mb{x}, \mb{s}_{obs}) &= \left[ \begin{matrix} \mb{f}_{N}(\mb{x}, \mb{s}_{obs}) \\ \mb{f}_{M} \end{matrix} \right]  \\
	 					      &=  \left[ \begin{matrix} (B^{r}_N)^{-1} (\mb{x} - B^{r}_{M} \mb{s}_{obs} ) \\ \mb{s}_{obs} \end{matrix} \right]. \label{f-star-solution}
\end{align}
Next we must show that the solution to this is in fact the correct spanning tree on the graph.
\begin{thm}
If $\mc{M}$ satisfies the condition in Theorem \ref{thm:spanning_tree_identifiability}, the solution vector $\mb{f}^{\star}(\mb{x}, \mb{s}_{obs})$ encodes spanning tree $\mc{T}$. \label{thm-f-star-correct-spanning-tree}
\end{thm}

The search over the set of spanning trees can be replaced by solving a set of linear equations \eqref{relaxed-flow-solution}.
This reduction is not only useful for a fast deterministic detector, but is used to formulate a flow based approximate ML detector.

\section{Stochastic Case} 
\label{section-stochastic-case}

This section presents the structure of a combinatorial ML detector as well as two approximate ML detection algorithms.

\subsection{MAP Detector Structure }
\label{subsection-MAP Detector Structure }

Given $G$, we can represent the observed flow as a linear function of consumption: 
\begin{align}  
\mb{s}(\mc{T}, \mb{x}) = \Gamma(\mc{T}, \mc{M})\mb{x}
\end{align}  
where
\begin{align}
\Gamma(\mc{T}, \mc{M}) &= A_{\mc{M}} B^{r, -1}_{\mc{T}}  \label{Gamma_Incidence_Relations}.
\end{align}
The subscript $B_{\mc{T}}$ indicates the incidence matrix corresponding to tree $\mc{T}$. 
Shorthand, $\Gamma_i$ denotes $\Gamma(\mc{T}_i, \mc{M})$ given a fixed sensor placement.

Next using the island graph representation, the general MAP detector in  \eqref{eq-switch-MAP-detector} can be evaluated for observed edge flows $\mb{s}_{obs}$, load forecasts $\hat{\mb{x}}$ and candidate spanning tree $\mc{T}$:
\begin{align}
\hat{\mc{T}} &= \underset{ \mc{T} \in \mbb{T}  }{\arg\max} \Pr\left( \mc{T}~|~\mb{s},~\mb{\hat{x}} \right)                                                                                                     	     \label{eq:map-formulation-line1} \\
 		   &=  \underset{ \mc{T} \in \mbb{T} }{\arg\max} \frac{ \Pr\left( \mb{s},~\mb{\hat{x}} ~ | ~ \mc{T} \right) \Pr\left( \mc{T} \right) }{ \Pr\left( \mb{s},~\mb{\hat{x}} \right)  }  \label{eq:map-formulation-line2} \\
    		   &=  \underset{ \mc{T} \in \mbb{T}  }{\arg\max} \Pr\left( \mb{s},~\mb{\hat{x}} ~ | ~ \mc{T} \right) \Pr\left( \mc{T}\right)                                                                          \label{eq:map-formulation-line3} \\
	           &= \underset{  \mc{T} \in \mbb{T} }{\arg\max} \Pr\left( \mb{s}~ | ~\mb{\hat{x}},~ \mc{T}\right)  \Pr\left( \mb{\hat{x}}~|~ \mc{T} \right) \Pr\left(\mc{T} \right)                 \label{eq:map-formulation-line4} \\
	           &= \underset{ \mc{T} \in \mbb{T}  }{\arg\max} \Pr\left( \mb{s}~ | ~\mb{\hat{x}},~ \mc{T} \right)  \Pr\left( \mb{\hat{x}} \right) \Pr\left( \mc{T} \right)                                \label{eq:map-formulation-line5} \\
                   & = \underset{  \mc{T} \in \mbb{T}  }{\arg\max} \Pr\left( \mb{s}~ |~\mb{\hat{x}},~\mc{T} \right).                                                                                                             \label{eq:map-formulation-line6}
\end{align}

Lines \eqref{eq:map-formulation-line1} - \eqref{eq:map-formulation-line4} convert the MAP detector to a likelihood detector with prior weights.
Line \eqref{eq:map-formulation-line5} conditions on the load forecast $\mb{\hat{x}}$.
Since $\mb{\hat{x}}$ does not depend on the outage hypothesis (only $\mb{s}$ does), the term can be removed leading to \eqref{eq:map-formulation-line6}.
Additionally, we assume a uniform prior over all hypotheses, however this does not have to be the case.
Therefore it is equivalent to a maximum likelihood estimate of the observed flow given a hypothesized tree and the load forecasts.

Given the forecasted loads $\hat{\mb{x}}$, the true loads at each node are given as: $\mb{x} \sim N(\hat{\mb{x}}, \sigma^2 I)$.
Therefore, under a particular hypothesized spanning tree $\mc{T}$, the true flow would be distributed as:
\begin{align}
\mb{s}( \mc{T}_i,~\mb{x}) &= \Gamma_i \mb{x} 								\label{eq:stochastic-flow-obs-eq1}  \\
					&=\Gamma_i (\hat{\mb{x}} +  \epsilon )				\label{eq:stochastic-flow-obs-eq2}  \\
					& = \mb{s}(\mc{T}_i, \mb{\hat{x}}) + \epsilon_{s, i}	        \label{eq:stochastic-flow-obs-eq3} \\
					& \sim N( \mb{s}(\mc{T}_i, \mb{\hat{x}}),~\Sigma_{s, i} ).	\label{eq:stochastic-flow-obs-eq4} 
\end{align}   
The observed flow $s_{obs}$ is $\mb{s}(\mc{T}_i, \mb{x})$, since it is the flow from tree $\mc{T}_i$ and true loads $\mb{x}$.
The term $\mb{s}(\mc{T}_i, \mb{\hat{x}})$ in   \eqref{eq:stochastic-flow-obs-eq3} indicates the theoretical flow that should be observed under the forecast of nodal consumption.
The error $\epsilon_{s, i}$ in  \eqref{eq:stochastic-flow-obs-eq3} is a zero mean multivariate Gaussian with covariance matrix 
$\Sigma_{s, i} = \sigma^2 A_{\mc{M}} B^{r, -1}_{\mc{T}} B^{r, -1, T}_{\mc{T}} A^{T}_{\mc{M}}$.
Given the distribution of what the flow should be, once an observation is given, we can perform maximum likelihood detection with
\begin{align}
\mc{T}  &= \underset{ \mc{T} \in \mbb{T}_{\tau} }{\arg\max} \left(  \mb{s}_{obs} - \mb{s}(\mc{T}, \mb{x} )  \right)^{T}  \Sigma^{-1}_{s, i}  \left(  \mb{s}_{obs} - \mb{s}(\mc{T}, \mb{x} )  \right).   \label{eq-combinatorial-MAP-detector} 
\end{align}
This detector, although optimal, requires enumeration of all spanning trees.
Since edge measurements in the island graph map to switches in the original network, the observation array and covariance matrix will be degenerate in that many zero's will be observed.
In such cases, the search space and likelihood function can be pruned and reduced in size.
We now present two approximate algorithms for solving \eqref{eq-combinatorial-MAP-detector}.

\subsection{Cycle Descent Approximate ML Detection}
\label{subsection-cycle-descent-approximate-MAP-Detection}
An approximate ML detector is based on generating single cycle edge exchanges $\Delta e_i = \{e_i \rightarrow e^{\prime}_i \}$ which iteratively maximizes the likelihood of the observations.
For every edge in $e_i$ in the co-tree of the current tree, an edge exchange $e^{\prime}_i$ is chosen to be an edge along the fundamental cycle of $e_i$.
Therefore, at every step an edge is chosen such that it corresponds to a hypothesized flow $\mb{s}(\mc{T}, \mb{x} )$ closer to the observed flow $\mb{s}_{obs}$.
The procedure is presented in Algorithm \ref{alg-cycle-descent-algorithm}.
\begin{algorithm}[h]   
\KwIn{ [1] Observed flows $\mb{s}$. \\
           \hspace{14.5mm} [2] Load Forecast $\hat{\mb{x}}$ and Error Covariance $\Sigma$ \\ 
           \hspace{14.5mm} [3] Graph $G$} 
\KwOut{MAP Detection Hypothesis $\mc{T}$}
\label{alg-cycle-descent-algorithm}
\textbf{// Find feasible start point. }  \\
$\mc{T}~\leftarrow$ {\bf feasible}-{\bf tree}$(\mb{s}_{obs},~\hat{\mb{x}},~G)$ \\
$FC \leftarrow {\bf generate}-{\bf fundamental}-{\bf cycle}\left( \mc{T} \right)$ 

\While { $\Delta \text{loglik} \neq 0$ } { 
	\For { $\mb{c}_k \in FC$ }{
		$\{ \text{loglik}, \mc{T} \} \leftarrow {\bf local}-{\bf update}( \mc{T}, \mb{c}_k ,~\mb{s},~\hat{\mb{x}} )$ \\
		$FC \leftarrow {\bf update}-{\bf cycles}\left( \mc{T}, FC \right)$ 
	}  
}
\caption{Cycle Descent Algorithm.}
\end{algorithm}

The cycle descent algorithm performs the following sub-tasks.

\noindent
\textit{{\bf feasible}-{\bf tree}} -
A feasible starting point $\mc{T}$ is chosen such that $\mbb{I}{ \{ \mb{s}_{obs} \neq 0 \} } =\mbb{I}{ \{ \mb{s}(\mc{T}, \mb{\hat{x}}) \neq 0 \} }$.
Note that if this is not the case, $\text{loglik} = -\infty$ and the procedure will fail.
This is computed with the following procedure.
\begin{enumerate}
\item Edges with measurements are weighted as follows: (1) edges measuring zero are weighted $0$; (2) edges measuring some non zero value are weighted $K$.
\item Remaining edges are assigned a very large weight ($\geq |E|~K$).
\item Maximum weight spanning tree is calculated on the weighted graph.
\end{enumerate}
This procedure will always produce at least one $\mc{T} \in \mbb{T}_{+}$, since we never choose an edge with zero weight.
This starting point may be very far from a optimal value, but will have a finite log likelihood.

\noindent
\textit{{\bf local}-{\bf update}} -
For a particular edge in the co-tree, $e_k \in G \setminus \mc{T}$, we have the fundamental cycle, $c_k$, produced by enumerating the single cycle formed from $\mc{T} + e_k$.
We then evaluated the objective,  \eqref{eq-combinatorial-MAP-detector}, with candidate trees $\mc{T}^{\prime} \leftarrow \mc{T} - e_k + e_j$, where $e_j \in E(c_k)$ and choose the maximum.

\noindent
\textit{{\bf update}-{\bf cycles} } - 
After each edge exchange operation, the cycles must be updated to reflect the exchanged edge.
A queue is maintained for the edge $e_k$ to be processed, where the elements are updated while maintaining the order of operation in $FC$.  
   
\subsubsection{Intuition of Cycle Descent Performance}

We can think of the likelihood function as $\Pr( \mb{s}~|~\mb{\hat{x}},~\mc{T} )$ as a function, $f(e_1, \hdots, e_{\mu})$, of the co-tree edges, where they must satisfy $E \setminus \{e_1, \hdots, e_{\mu} \}$ being a spanning tree.
The cycle descent algorithm assumes at every stage that
\begin{align}
f(e_1, \hdots, e_{\mu}) = \prod_{e_k \in c_k} f_{k}(e_k).
\end{align}
Therefore, taking the greedy choice is optimal.
To see why this is a good approximation, consider Figure \ref{fig:fc_edge_exchange} and Lemma \ref{lem:decouple-on-fc}, where in the noiseless case, only sensors on the fundamental cycle are effected by the candidate edge moving along the cycle.
All the other sensors not on the cycle are fully decoupled.
In general, this decoupling is not necessarily true, but the approximation is close and as will be shown in the numerical simulations, almost all of the spanning trees will have the same performance as the combinatorial method.
%
%
%
\subsection{Flow Based Approximate ML Detector}
\label{subsection-flow-based-APX-MAP-detector}

This section shows how the combinatorial detector can be reformulated in terms of a network flow based mixed integer quadratic program.
An alternative interpretation to this development is a hypothesis testing framework which is discussed in Appendix \ref{subsection-alternative-view-of-FMST}.

The combinatorial ML in \eqref{eq-combinatorial-MAP-detector} can be rewritten in terms of an estimated flow $\mb{f}$ and unknown spanning tree $\mc{T}$ constraint similar to \eqref{opt1}.
Consider the program,
\begin{align}
& \text{min}~ \frac{1}{2} (  B^{r}_{\mc{T}} \mb{f} - \mb{\hat{x}} ) \Sigma^{1}_{\mc{T}}(  B^{r}_{\mc{T}} \mb{f} - \mb{\hat{x}} )  - \frac{1}{2} \ln \left( \det( \Sigma^{2}_{\mc{T} } )\right) 	\label{opt2} \tag{OPT-2}  	\\
& \text{s.t.} 																																\nonumber 			\\
& ~~~~~ A_{M, +} \mb{f} = \mb{s}_{obs,+}      																										\label{opt2-flow-obs}  	\\
& ~~~~~ \mc{T} \in \mbb{T}_{+}        																												\label{opt2-sp-tree}
\end{align}
and the following equivalence:

\begin{thm}
\label{thm-comb-flow-rewrite}
The combinatorial detector in \eqref{eq-combinatorial-MAP-detector} is equivalent to \eqref{opt2}.
\end{thm}

Like \eqref{opt1}, we must solve for an estimated flow $\mb{f}$ and discrete topology $\mc{T}$.
Here, matrices  $A_{M, +}$, $\Sigma^{1}_{\mc{T}}$, and $\Sigma^{2}_{\mc{T} }$ and the search space $\mbb{T}_{+}$ depend on the binary array $\mbb{I}{\{ \mb{s}_{obs} \neq 0 \} }$.
Matrices $\Sigma^{1}_{\mc{T}}$, and $\Sigma^{2}_{ \mc{T} }$ and $B_{\mc{T}}$ depend on the candidate spanning tree.
Eq. \eqref{opt2-flow-obs} is the observation constraints corresponding to the true flow observation.
Eq. \eqref{opt2-sp-tree} constrains the search space to all spanning trees which lead to flow observations which satisfy the $\mbb{I}{\{ \mb{s}_{obs} \neq 0 \} }$ observations.

Intuitively, \eqref{opt1} and \eqref{opt2} are very similar.
However, it is not possible to find a flow satisfying $B_{\mc{T}} \mb{f} = \mb{\hat{x}}$ and $A_{\mc{M}} \mb{f} = \mb{s}_{obs}$, due to the error in the nodal measurements.
A clear alternative is to find a flow and tree which minimizes a distance measure between the predicted nodal measurements $\mb{\hat{x}}$ and $B_{\mc{T}} \mb{f}$.

\eqref{opt2} is still difficult, since we must search over $\mc{T}$. 
This can be approximated and solved in a much easier fashion by the following coordinate descent style solution.
Recall in coordinate descent, a minimization over two sets of variables, for example $x$ and $y$, will be performed once over variable $x$, then over variable $y$.
That is, 
$\{x^{\star}, y^{\star} \} = \arg\min_{ x \in \mb{X}, y \in \mb{Y} }g(x, y)$ 
is approximated by $\{x^{\star} \} = \arg\min_{ x \in \mb{X} }~g(x, y_{0})$ and $\{y^{\star} \} = \arg\min_{ y \in \mb{Y} }g(x^{\star}, y)$.

To perform a coordinate descent optimization over $\mb{f}$ then $\mc{T}$, we perform the following:
\begin{enumerate}
\item  Setting $\mc{T}_{0}$ to the fully connected graph, and solve for the optimal flow, $\mb{f}^{\star}$, to minimize the objective.
\item  Using the solution $\mb{f}^{\star}$, minimize an approximate form of the original objective which leads to an efficient solution.
\end{enumerate}

\subsubsection{Step 1: Solving over flow}

Relaxing the spanning tree constraint makes \eqref{opt2} ill defined since the choice of $\Sigma^{1}_{\mc{T}} $ and $\Sigma^{2}_{\mc{T}} $ is undefined.
Additionally, these terms are difficult to evaluate.
A clear alternative is to just remove these reweighing matrices altogether, and aim to simply solve for the least square error in the following form
\begin{align}
\mc{\hat{T}}& = \arg\min \frac{1}{2} \|  B \mb{f} - \mb{\hat{x}}  \|^2  	\label{opt2-reduced} \tag{OPT-3}  	\\
& \text{s.t.} 											\nonumber 			\\
& ~~~~~ A_{M, +} \mb{f} = \mb{s}_{obs,+}      					\label{opt2-reduced-flow-obs}  	\\
& ~~~~~ \mc{T} \in \mbb{T}_{+}.        						\label{opt2-reduced-sp-tree}
\end{align}

Solving the relaxed objective leads to 
\begin{align}
\mb{f}_{obs} = \left[ \begin{matrix}  B^{r, -1}_{N}( \mb{\hat{x}} - B^{r}_{M}\mb{s}_{obs} ) \\ \mb{s}_{obs} \end{matrix} \right].
\end{align}

This is identical to the matrix inversion based detector, except the forecast $\mb{\hat{x}}$ is used instead of the true value.
This is the so-called 'noisy flow' solution, which is discussed in Appendix \ref{subsection-alternative-view-of-FMST}, and final objective value is $0$ since $B\mb{f}_{obs} = \mb{\hat{x}}$.
\subsubsection{Step 2: Solving over flow}

Now optimizing over the second coordinate, leads to
\begin{align}
\text{OPT}(\mb{f}_{obs}) &= \underset{ \mc{T} \in \mbb{T}^{+} } {\min} \frac{1}{2} \|  B_{\mc{T}} \mb{f}_{obs} - \mb{\hat{x}}  \|^{2} 			\\	    	   
		           		&= \underset{ \mc{T} \in \mbb{T}^{+} } {\min} \frac{1}{2} \|  B_{\mc{T}} \mb{f}_{obs} - B\mb{f}_{obs} + B\mb{f}_{obs} - \mb{\hat{x}}  \|^{2} \\	    	   
		           		&= \underset{ \mc{T} \in \mbb{T}^{+} } {\min} \frac{1}{2} \|  B_{\mc{T}} \mb{f}_{obs} - B\mb{f}_{obs} \|^{2} 		\\	    	   		   
			   		&= \underset{ \mc{T} \in \mbb{T}^{+} } {\min} \frac{1}{2} \|  B_{G \setminus \mc{T} } \mb{f}_{obs}  \|^{2}.   	\label{OPT_F_L4}   	   		   		        
\end{align}

A close approximation to $\text{OPT}(\mb{f}_{obs})$ is the minimum spanning tree solution over the negative absolute weights: $\text{MST}(-|\mb{f}_{obs}|)$.
This is equivalent to finding the tree with the maximum edge weights in $|\mb{f}_{obs}|$
For this, we have the following bounds relating the approximate solution and the desired $\text{OPT}(\mb{f}_{obs})$.

\begin{thm}
\label{thm:MST-OPT-BOUNDS}
For any flow vector $\mb{f}_{obs}$, $\text{OPT}(\mb{f}_{obs}) \leq  \text{MST}(-|\mb{f}_{obs}|^2)$.
\end{thm}

In standard approximate algorithm analysis \cite{Williamson2011}, our approximate technique should bound the optimal solution from above and below by a constant factor which does not depend on the problem instance.
Therefore this is not complete analysis of an approximation algorithm. 
However, it shows why solving a minimum spanning tree over the ``noisy-flows'' leads to a decent approximate solution and leads to the following approximation algorithm to the combinatorial MAP detector.
This is shown experimentally in Section \ref{subsection-approximate-ML-detector}.

The procedure is described fully in Algorithm \ref{alg-noisy-flow}.
\begin{algorithm}[h]   

\KwIn{ [1] Observed flows $\mb{s}_{obs}$. \\
           \hspace{14.5mm} [2] Load Forecast $\hat{\mb{x}}$. \\ 
           \hspace{14.5mm} [3] Graph $G$} 
\KwOut{MAP Detector Output $\mc{T}$}
\label{alg-noisy-flow}
Evaluate the empirical flow $\mb{f}(\hat{\mb{x}}, \mb{s}_{obs} )$ via \eqref{eq:noisy-flow-line1}. \\

Compute the minimum weight spanning tree solution on graph $G$ where edges are weighted with $-|\mb{f}(\hat{\mb{x}}, \mb{s}_{obs} )|$. \\
\caption{Flow Based Approximate ML Detector}
\end{algorithm}

%
%
%

\section{Numerical Experiment}
\label{section-Numerical-Experiment}

This section presents the following analysis: (1) deterministic detector under various situations; (2) stochastic detection problem for the combinatorial and approximate detectors; (3) numerical results in sensor placement in a stochastic case; (4) analysis of the IEEE 123 Test system.

Two error metrics which are used throughout the numerical section are
\begin{itemize}
\item mean missed detection error over all possible spanning trees 
\begin{align}
g_1(\mc{M}) = \sum_{\mc{T} \in \mbb{T} } \Pr( \mc{T} ) \Pr( \hat{\mc{T}} \neq \mc{T} | \mc{T} ; \mc{M} ); \label{eq-mean-error} 
\end{align} 
\item maximum missed detection error over all possible spanning trees
\begin{align}
g_2(\mc{M}) = \max_{ \mc{T} \in \mbb{T} } \Pr( \hat{\mc{T}} \neq \mc{T} | \mc{T} ; \mc{M} ).  \label{eq-max-error}  
\end{align}  
\end{itemize}

\subsection{Deterministic Placement}
\label{subsection-Deterministic-Placement}

We test the placement problem on a set of planar graphs, shown in Figure \ref{fig:test-graphs}.
In both graphs, a single vertex is designated as the source which is as the top most horizontal lines.
Graph $G_1$ has $v_{\text{root}} = v_4$ and $G_2$ has $v_{\text{root}} = v_1$.

\begin{figure}[h]
\centering
\subfigure[][]{ 
	\label{fig:test_graph_1}
	\includegraphics[scale=0.33]{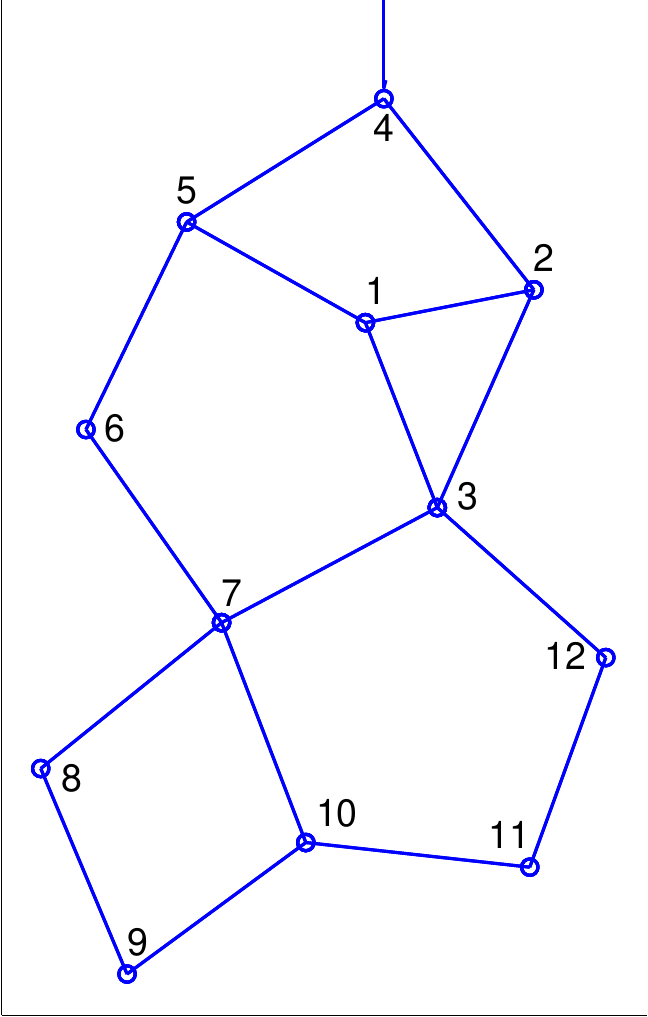}
}
\hspace{-5mm}
\subfigure[][]{     
	\label{fig:test_graph_2}
	\includegraphics[scale=0.25]{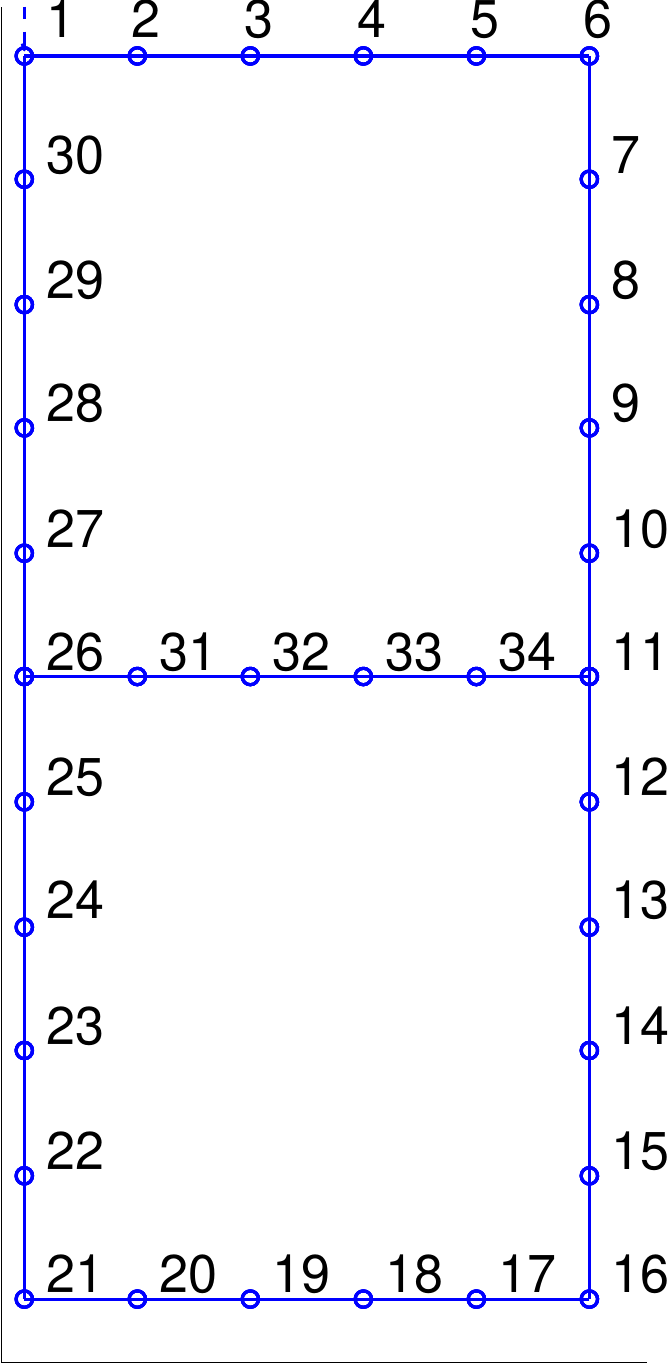}
}
\caption[Test Graph $G_1$ and $G_2$]{
Sample graphs used in various experiments  \subref{fig:test_graph_1} $G_1$ and \subref{fig:test_graph_2} $G_2$.
}
\label{fig:test-graphs}
\end{figure}

\begin{table}[h]
\centering
\caption{Deterministic Topology Detection.}    
\begin{tabular}{@{}cccccccc@{}}
\toprule
  ~            && $\mu$   & $|\mbb{T}|$  &                &    $\epsilon^{\prime}$ &                    & $|E|/\mu$  \\ 
    ~          &&                  &                      &  mean     &              std.              & max/min     &                       \\ 
\cmidrule{1-1}    \cmidrule{3-8}   
$G_1$     &&          5      &       391   	    & $56 .9$   &               $26.7$       &  $299/6$     &  2.8   \\ 
$G_2$     &&          5      &       830         & $139.3$  &               $72.2$       &  $185/10$   &  3.2   \\ 
\bottomrule  
\end{tabular}
\label{tab:deterministic_topology_detection}
\end{table}    

To test the placement problem, we enumerate the set of spanning trees for each of the graphs.
The method relies on the backtracking method developed in \cite{Gabow1978}.
The simulation was implemented in MATLAB and deemed correct by checking that each spanning tree was unique and the number of test trees corresponded to those calculated from the matrix-tree theorem \cite{Diestel2000}.
The theorem allows us to compute the number of unique spanning trees without explicit enumeration.
The number of spanning trees is $|\mbb{T}| = \det(L_v)$ where $L_v$ is the $v$ minor of the Laplacian matrix with the result being invariant to $v$.

For the graphs in Figure \ref{fig:test-graphs}, the graph statistics and experiment results are shown in Table \ref{tab:deterministic_topology_detection}.
We evaluate the experimental error rate
\begin{align}
\epsilon = \frac{1}{N(N-1)} \sum_{i \neq j} \mbb{I}\{  \mb{s}(\mc{\hat{T}}_i, \mb{x}) \neq \mb{s}(\mc{T}_i, \mb{x})  \}.
\end{align}
From Theorem \ref{thm:spanning_tree_identifiability}, the missed detection error must be zero. 
The computed $\epsilon$ was zero in both cases, as was expected.

We evaluate the output according to an arbitrary input because we would like to compare it to the case where only magnitude and not direction is measured.
This is a common type of power system measurement as discussed in Section \ref{subsubsection-spanning-tree-detection-without-flow-direction}.
In this case, we evaluate $\epsilon^{\prime}$ which now compares $|\mb{s}(\mc{\hat{T}}_i, \mb{x})| \neq |\mb{s}(\mc{T}_i, \mb{x})|$ instead.
The computed values for $\epsilon^{\prime}$ are shown in Table \ref{tab:deterministic_topology_detection}.  
We evaluate each valid placement in $\mbb{M}$ to illustrate the importance of flow direction.
The value reported in Table \ref{tab:deterministic_topology_detection} is the mean missed detection error $\pm$ the standard deviation.
This verifies that different placements result in different unsigned missed detection rates.
We see that if the direction of flow is not known, around $10\%$ of the spanning trees are indistinguishable on average.

\subsection{MAP Detection Performance}
\label{subsection-MAP-Detection-Performance}
  
\begin{figure}[h]
\hspace{-5mm}
\centering
\subfigure[][]{ 
	\includegraphics[scale=0.25]{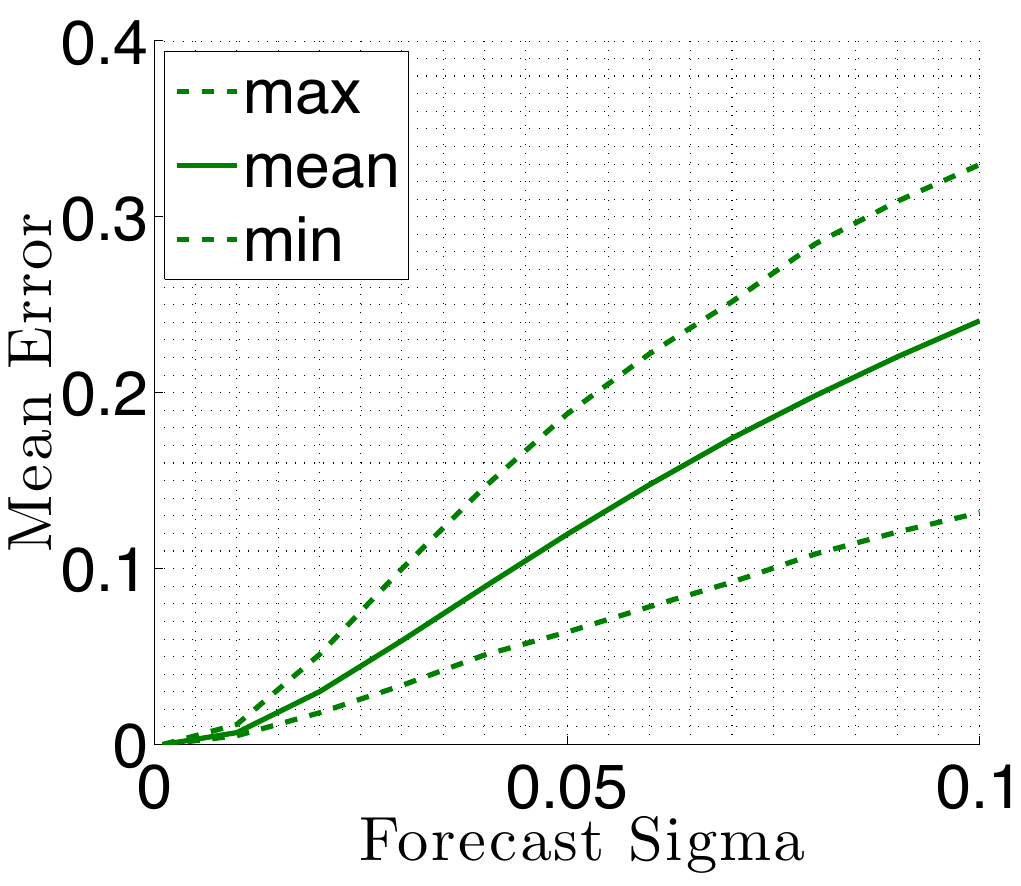}
	\label{fig:MAP-detector-performance-G2}
}
\hspace{-5mm}
\subfigure[][]{   
	\includegraphics[scale=0.25]{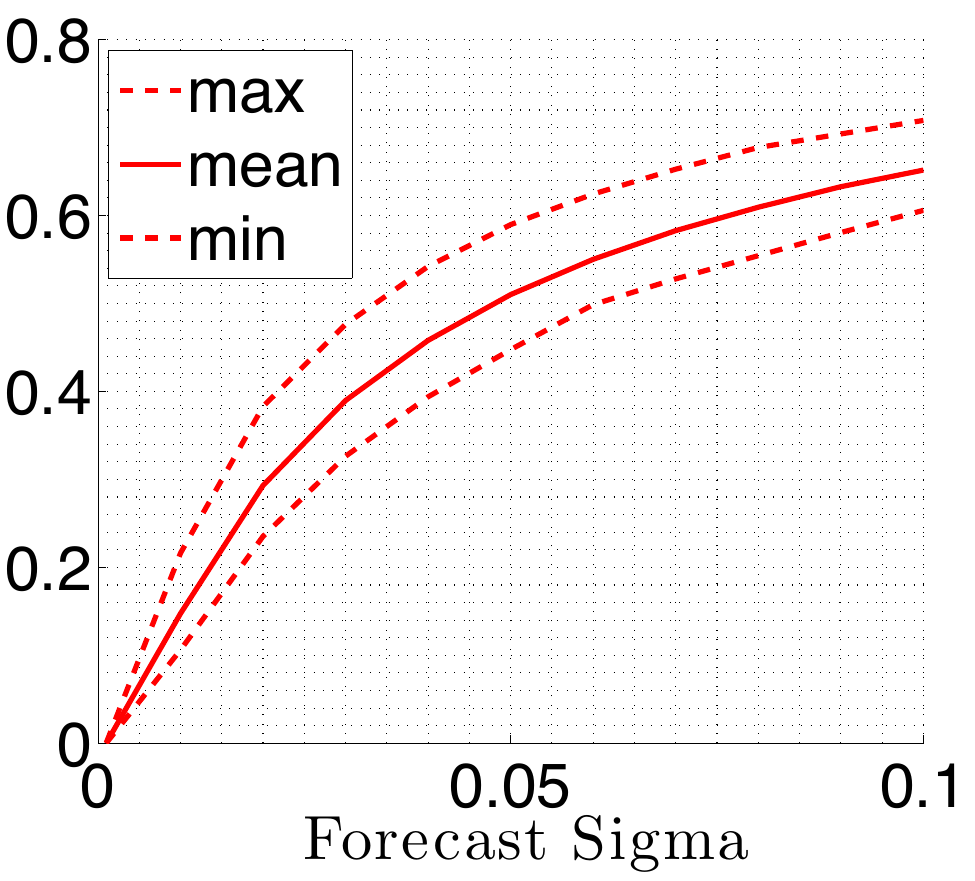}
	\label{fig:MAP-detector-performance-G3}
}
\caption[MAP Detector Performance]{  
Mean missed detection error for the graphs G1 (\subref{fig:MAP-detector-performance-G2}), G2 (\subref{fig:MAP-detector-performance-G3} with respect to $\sigma$.
}
\end{figure}

The performance of the ML detector is evaluated for each of the graphs in Figure  \ref{fig:MAP-detector-performance-G2}, \ref{fig:MAP-detector-performance-G3}.
The one shot detector performance is evaluated with a uniform load mean of $\mu_i = 1$ and forecast error of $\sigma$.
The figures show the mean missed detection error over all hypotheses with respect to $\sigma$.
  
A number of important observations can be seen from this analysis.
Different graphs experience widely different behavior.  
For example $G_1$ has many very short cycles where for any given spanning tree, multiples sensors will see multiple zeros while $G_2$ has only two cycles with high edge count per cycle.
The observations of zeros limits the number of candidate spanning trees that must be considered for the detector thereby pruning out many candidates.

The sensor placement has a dramatic impact on the mean missed detection error.
This is slightly counterintuitive, since a single placement which maps to a single spanning tree must correctly decode all spanning trees with low error.
A symmetry between placement and tree's would make one suspect that the missed detection error should not depend on any single placement.
Within a graph, it is observed that the the mean length of all fundamental cycles associated with a placement is slightly correlated with the mean error.
   
\subsection{Approximate ML Detector}   
\label{subsection-approximate-ML-detector}

\begin{figure}[h]
\hspace{-5mm}
\centering
\subfigure[][]{ 
\label{fig:APX-MAP-performance-G2}
\includegraphics[scale=0.25]{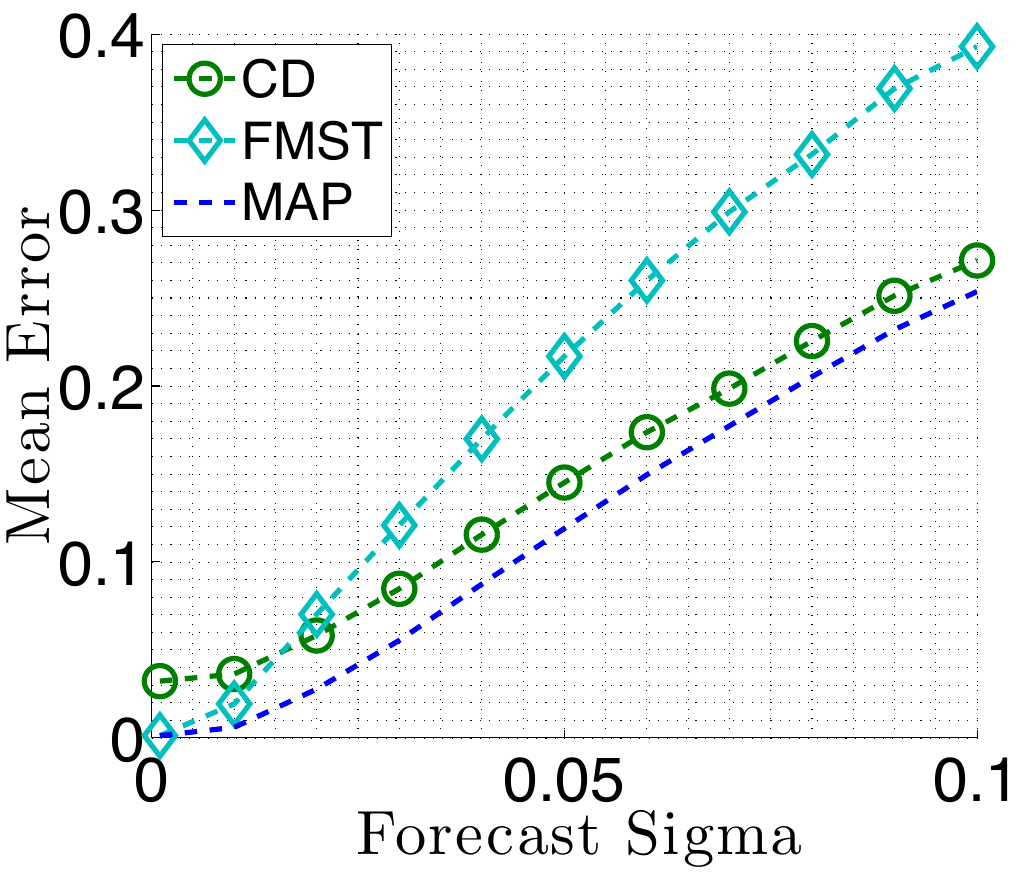}
}
\hspace{-5mm}
\subfigure[][]{   
\label{fig:APX-MAP-performance-G3}
\includegraphics[scale=0.25]{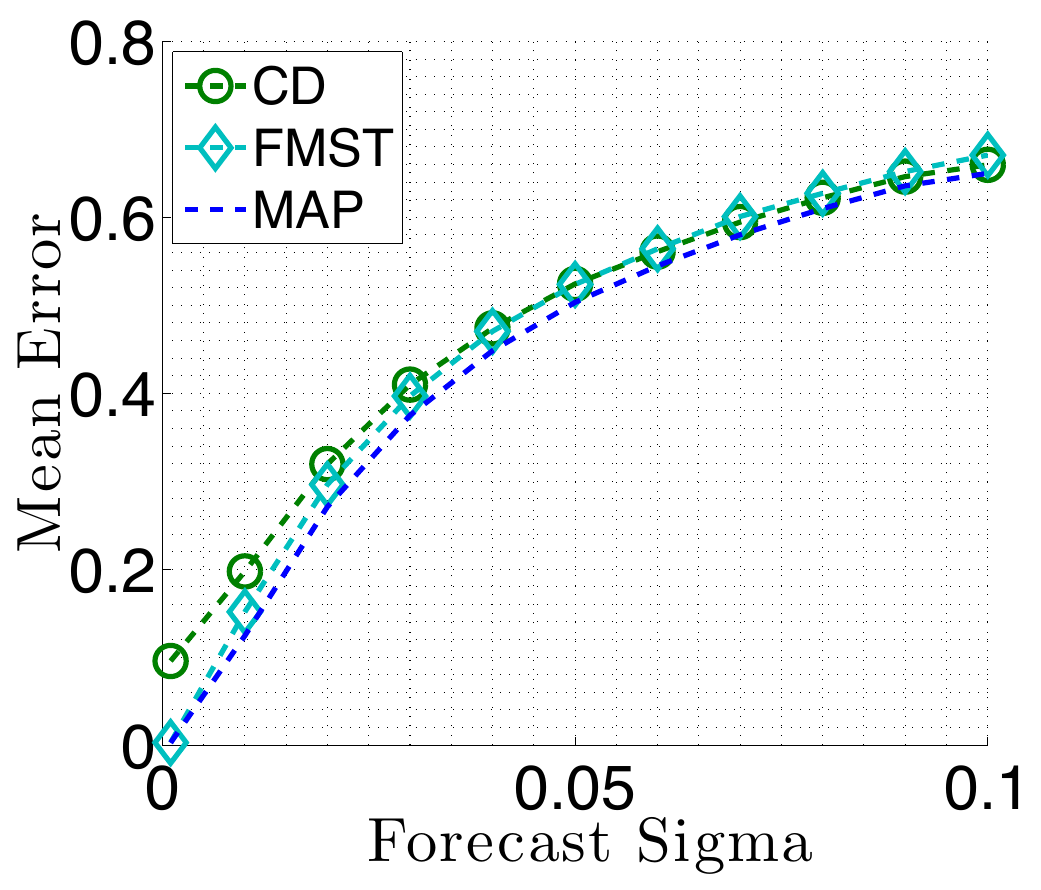}
}\caption[Approximate Detector Performance]{  
Comparison of mean missed detection error for optimal MAP detector, flow based approximate detector (FMST) and cycle descent detector (CD) for the graphs $G_1$ and $G_2$ with respect to $\sigma$.
}
\label{fig:MAP-detector-performance}
\end{figure}  

The two approximate MAP algorithms are tested on $G_1$ and $G_2$, where the performance are shown in Figure \ref{fig:APX-MAP-performance-G2}-\ref{fig:APX-MAP-performance-G3} respectively.

\subsubsection{Flow-Maximum Weighted Spanning Tree (FMST)}
In the near noiseless case, the flow based approximate ML detector performs identically to the combinatorial map detector, in both $G_1$ and $G_2$.
This is because the 'noisy-flow' values are very close to their correct values of zero.
In the high noise case, the algorithm fails worst in the high cycle count graph $G_1$, where the maximum spanning tree graphs very rarely match with the maximum likelihood output.
In $G_2$, however the two are nearly identical.

\subsubsection{Cycle Descent Algorithm}
The cycle descent algorithm has a very different performance than the FSMT algorithm.
For both graphs, the performance is similar.
For a small subset of trees in $G_1$, the algorithm always fail regardless of SNR.
For the remaining spanning trees, the algorithm converges to the optimal detector output.
In the case of $G_1$, only $4\%$ of the trees lead to a failure of the algorithm, for the used placement.

In simulation, it is verified that for the $4\%$ of cases which failed, the detector output corresponded to a tree which mapped to a different cycle basis as that of the correct tree.
Therefore, if the greedy algorithm finds a tree within the same cycle basis, it will find the correct solution (or combinatorial ML solution).
We suspect that this can avenue of investigation can lead to sub $O(|\mbb{T}|)$ optimal MAP detector, instead of an approximate technique.

\subsection{Sensor Placement in Stochastic Case}   
\label{subsection-Stochastic-Sensor-Placement}

\begin{figure}[h]
\hspace{-4mm}
\subfigure[][]{ 
	\label{fig:greedy_max_error_performance}
	\includegraphics[scale=0.22]{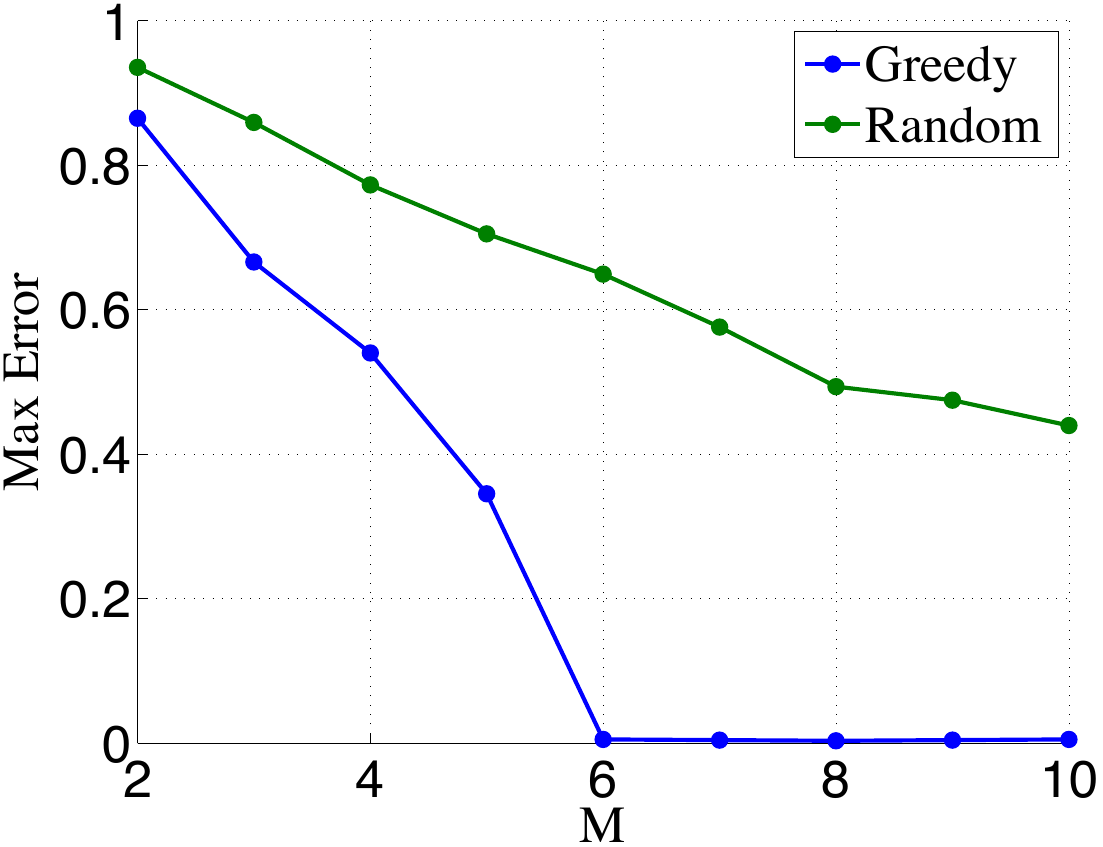}
}
\subfigure[][]{ 
	\label{fig:greedy_mean_error_performance}
	\includegraphics[scale=0.22]{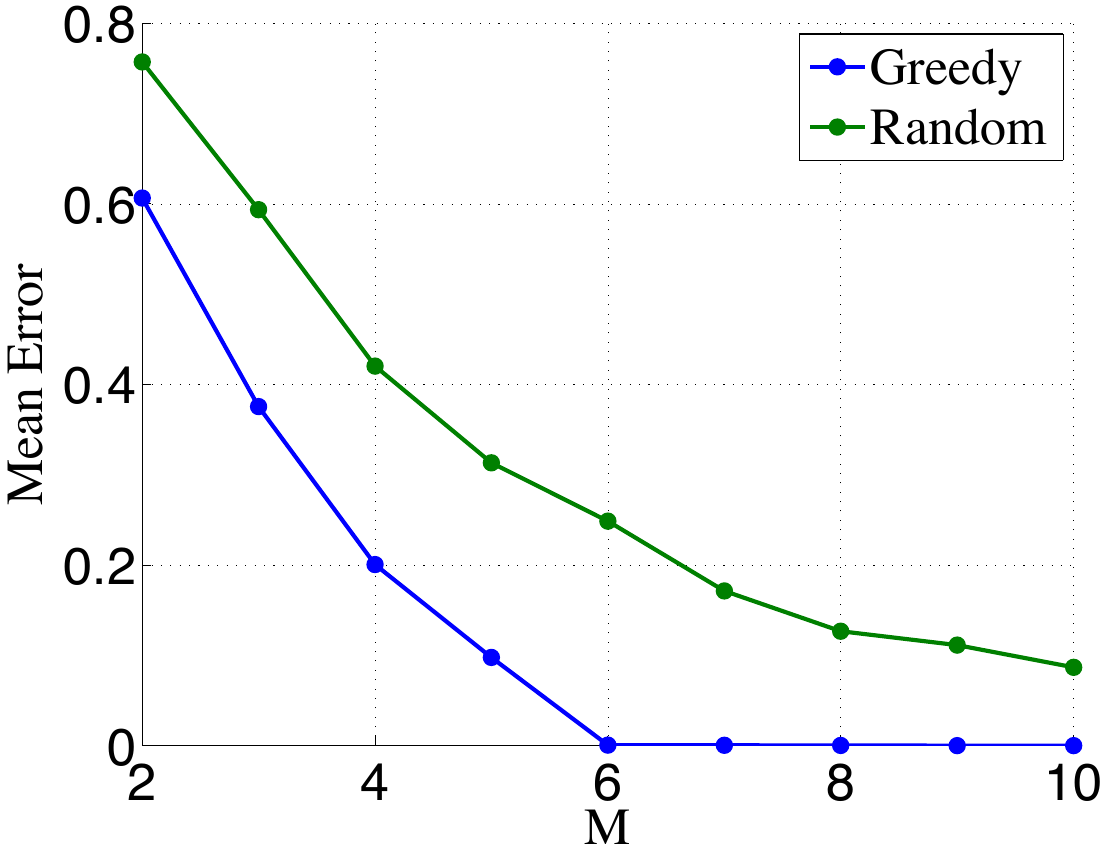}
}
\caption[Greedy Placement Performance]{
Performance of the greedy placement in Algorithm \ref{algorithm-greedy-sensor-placement} is shown for $G_3$. 
With the max-error objective \subref{fig:greedy_max_error_performance} and mean error \subref{fig:greedy_mean_error_performance}.
A random placement is shown for each $|\mc{M}|$ as comparison, with an evaluation size of $100$ placements.
The mean performance over the 100 samples is indicated.}
\label{fig:greedy_max_mean_error_G3}
\end{figure}

\begin{algorithm}[h]
\KwIn{[1] Graph $G$ \\ \hspace{10mm}
          [2] Nominal Load Statistics $\mathbf{L}$, $\mathbf{\Sigma}$ \\ \hspace{10mm}
          [3] Maximum Sensors M} 
\KwOut{Greedy Placement $\mc{M}^{g}$}
$\mc{M}^{g} \leftarrow \underset{ \mc{M} \in \mbb{M} }{\arg\min} ~~  g_{i}(\mc{M}) $ \label{placement-initialize} \\
\While { $| \mc{M}^{g} | \leq M$ } {
	$e^{\star} \leftarrow \underset{ e \in E \setminus \mc{M}^{g} }{\arg\min}~~g_{i}( \mc{M}^{g} \cup e )$ \label{placement-greedy-choice} \\
	$\mc{M}^{g} \leftarrow \mc{M}^{g} \cup \{ e^{\star} \}$ \label{placement-append-best-choice}
}
\label{algorithm-greedy-sensor-placement} 
\caption{Greedy Sensor Placement}
\end{algorithm}

The algorithm is tested on Graph $G_1$ with each node having an identical $\mu_i = 1$, $\sigma = 0.1$.
For comparison, we evaluate a maximum of 100 randomly allocated placements for each size: $\min\left(100,  {|E| \choose |\mc{M}|} \right)$.
The performance is indicated in Figure \ref{fig:greedy_max_mean_error_G3} for both metrics.
{\color{black} The graphs indicate a clear improvement as opposed to a randomized placement.}
For the mean error metric, $|\mc{M}^{g}| \geq 6$ has an error rate less than $0.005$ which is a sensor density of $17 \%$.
A randomized method has much poorer performance on average.
The results are much worse in the max error case, which is expected.
For very large sensor densities, the maximum error is still quite high.
For the max error metric, $|\mc{M}^{g}| \geq 6$ has an error rate less than $0.005$ which is a sensor density of $17 \%$.

\subsection{Objective Submodularity Counterexample}
\label{subsection-Submodularity-Counterexample}

Sub modularity is a property commonly exploited in many combinatorial optimization problems (see \cite{nemhauser1988}) for more details.
It is useful since it guarantees that a greedy algorithm is within a factor of $(1 - \frac{1}{e})$ of the optimal value.
\begin{define} 
For every $\mc{M} \subset \mc{M}^{\prime} \subset E$ we have $\forall e \in E \setminus B$
$g( \mc{M} \cup \{e\} ) - g(\mc{M})$ $\leq$ $g( \mc{M}^{\prime} \cup \{e\} ) - g(\mc{M}^{\prime})$.
\end{define}

The \textit{relative decrease in the objective function must be larger for the smaller set} under all subsets $\mc{M}$, $\mc{M}^{\prime}$ and additional element $e$.
Section \ref{subsection-Submodularity-Counterexample}, presents a numerical counterexample.

\begin{figure}[h]
\centering
\subfigure[][]{ 
	\includegraphics[scale=0.24]{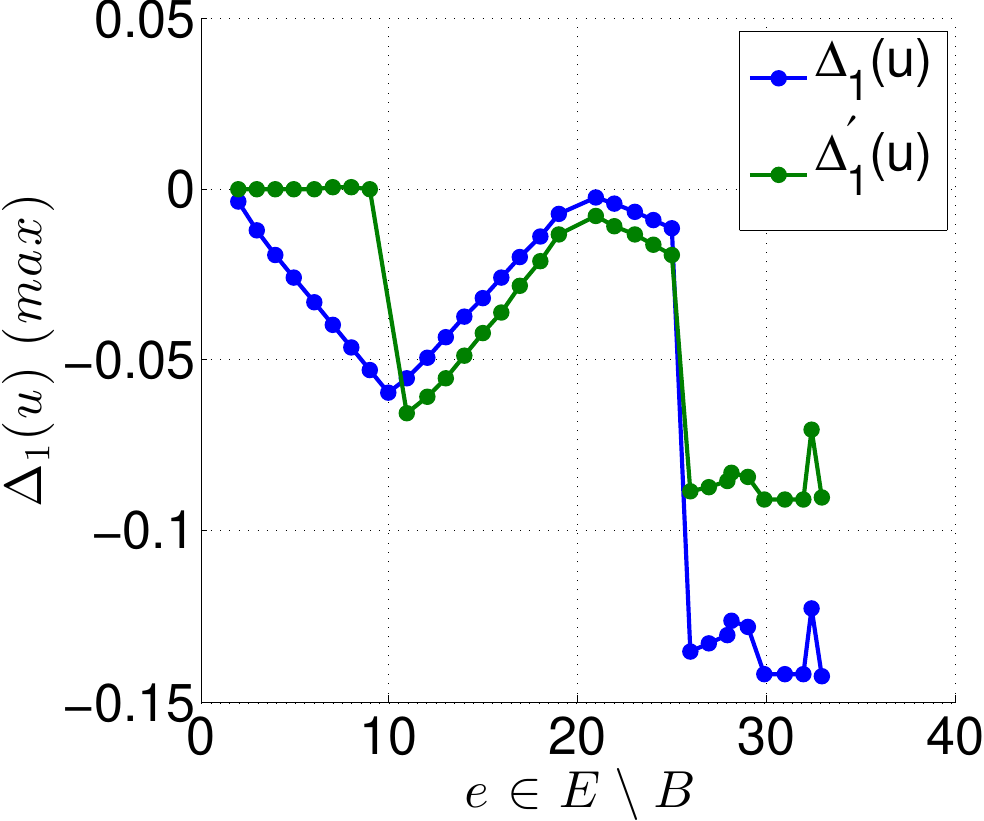}
	\label{fig:test_mean-submodularity-counterexample}  
}
\subfigure[][]{ 
	\includegraphics[scale=0.24]{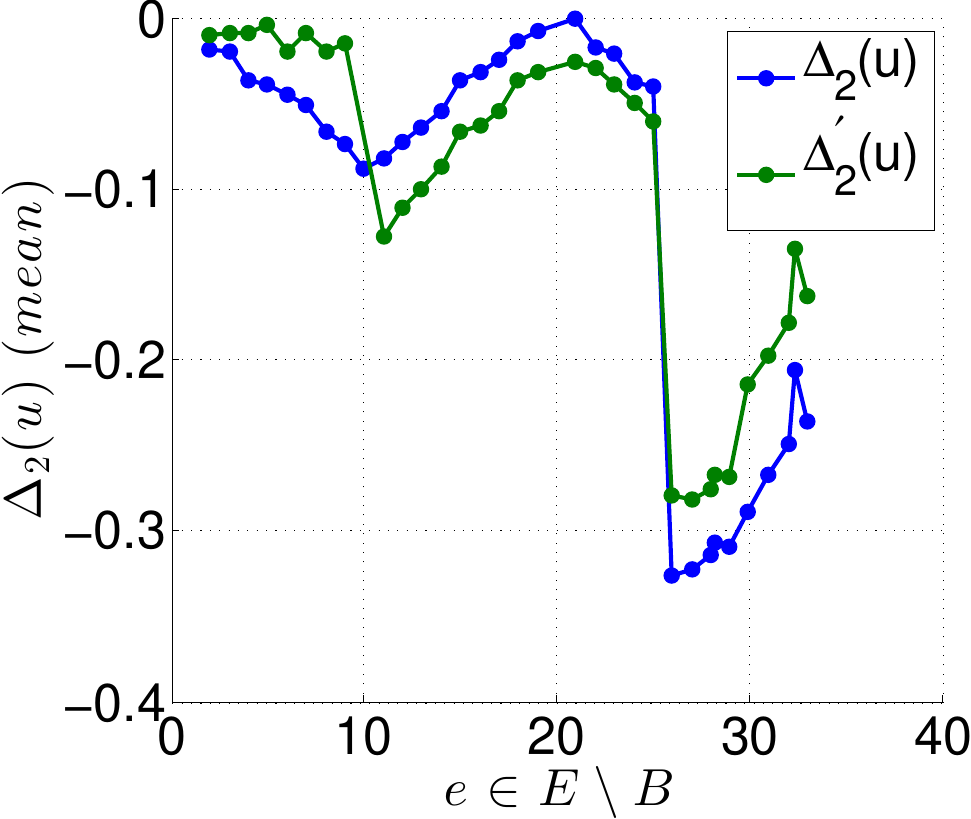}   	
	\label{fig:test_max-submodularity-counterexample}
}
\caption[Sub-modularity Counter Example]{
Counterexample for super modularity.  
For the objective to be super modular, $\Delta_{k}(e) (u) \leq \Delta^{\prime}_{k}(u)$ $\forall u \in E\setminus B$.  
The numerical experiment shows that this is not the case, from observing the two plots crossing in multiple points.  
}
\end{figure}

The sample graph $G_2$ is used to show a computable counterexample to sub modularity.
Consider the following sets $\mc{M}= \{(1,2), (20,21)\}$ and $\mc{M}^{\prime} = \{(1,2), (20,21), (10,11)\}$.
For the remaining allowable edge $e \in E \setminus \mc{M}^{\prime}$, we can compute the discrete dervative for each additional measurement where 
\begin{align}
\Delta_{k}(e) &= g_k( \mc{M} \cup \{e\} ) - g_k(\mc{M}) \\
\Delta^{\prime}_{k}(e) &= g_k( \mc{M}^{\prime} \cup \{e\} ) - g_k(\mc{M}^{\prime}).
\end{align}

For super modularity to hold, we must have $\Delta_{k}(e) \leq \Delta^{\prime}_{k}(e)$ for all $e$.
However, as the example shows, in a certain set of $u$ we have that $\Delta^{\prime}_{k}(e) < \Delta_{k}(e)$ for $k = 1,2$.

\subsection{123 Test Feeder and Robustness of Power Flow Measurements}
\label{subsection-123-Test-Feeder}

\begin{figure}[h]   
\hspace{-5mm}
\subfigure[][]{
\includegraphics[width=0.23\textwidth, height=0.23\textwidth]{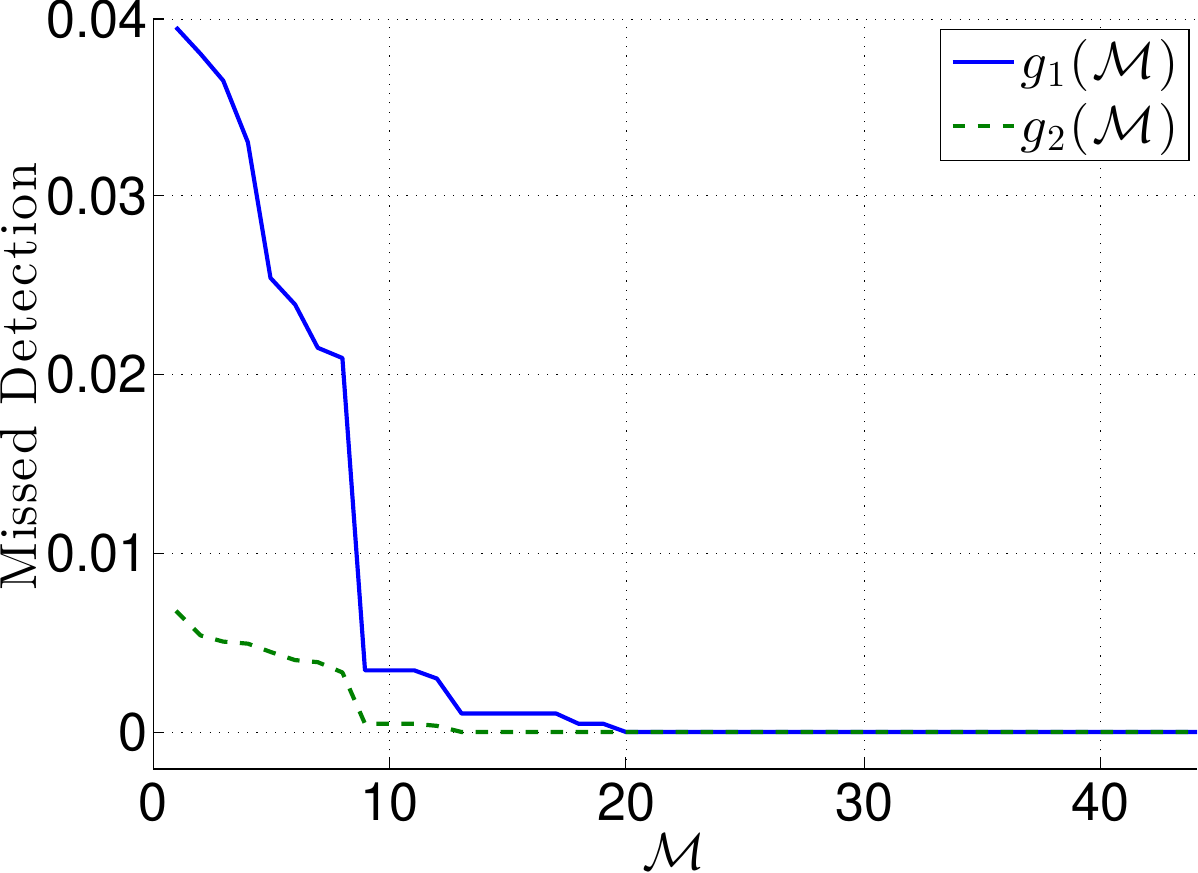}
\label{fig:IEEE123_placement}  
}
\hspace{-5mm}
\subfigure[][]{
\includegraphics[width=0.23\textwidth, height=0.23\textwidth]{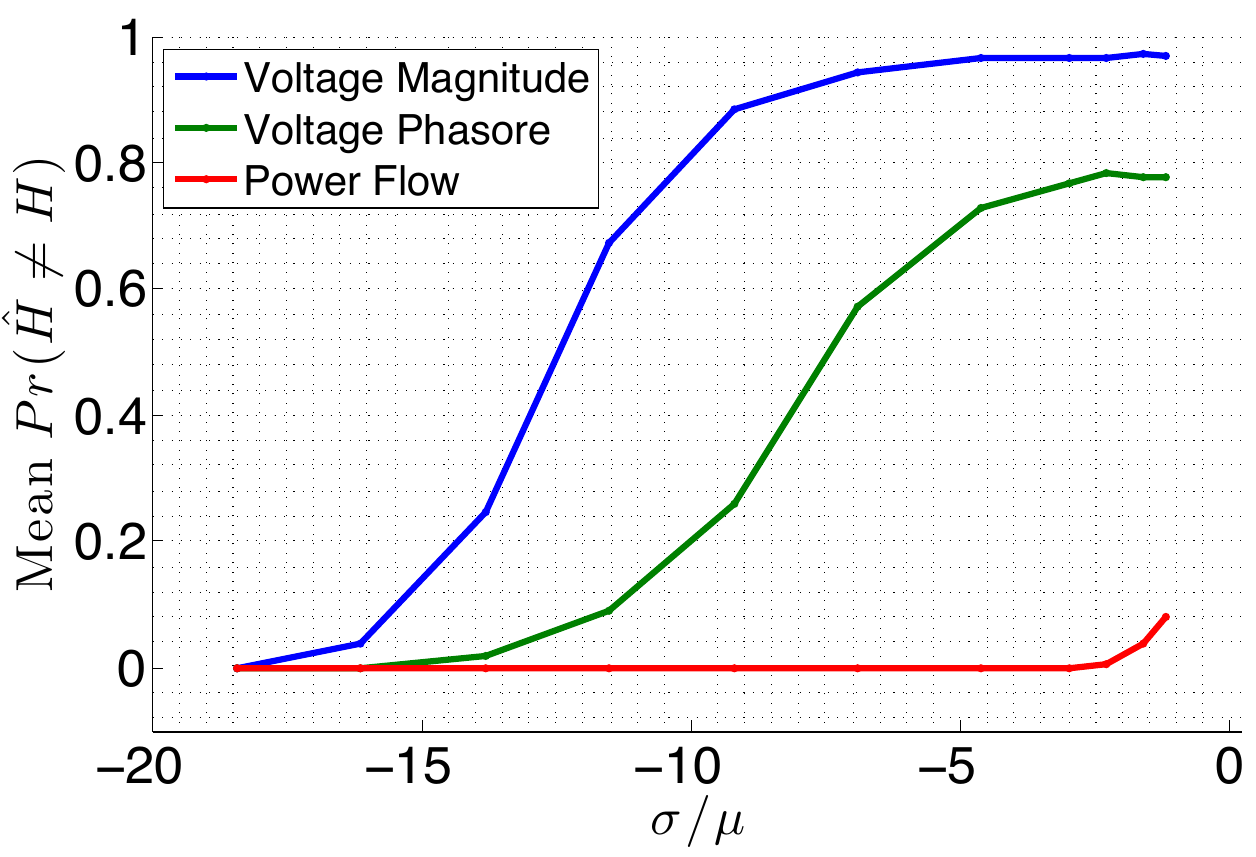}
\label{fig:IEEE123_voltage_vs_power_flow}
}
\subfigure[][]{
\hspace{-5mm}
\includegraphics[scale=0.4]{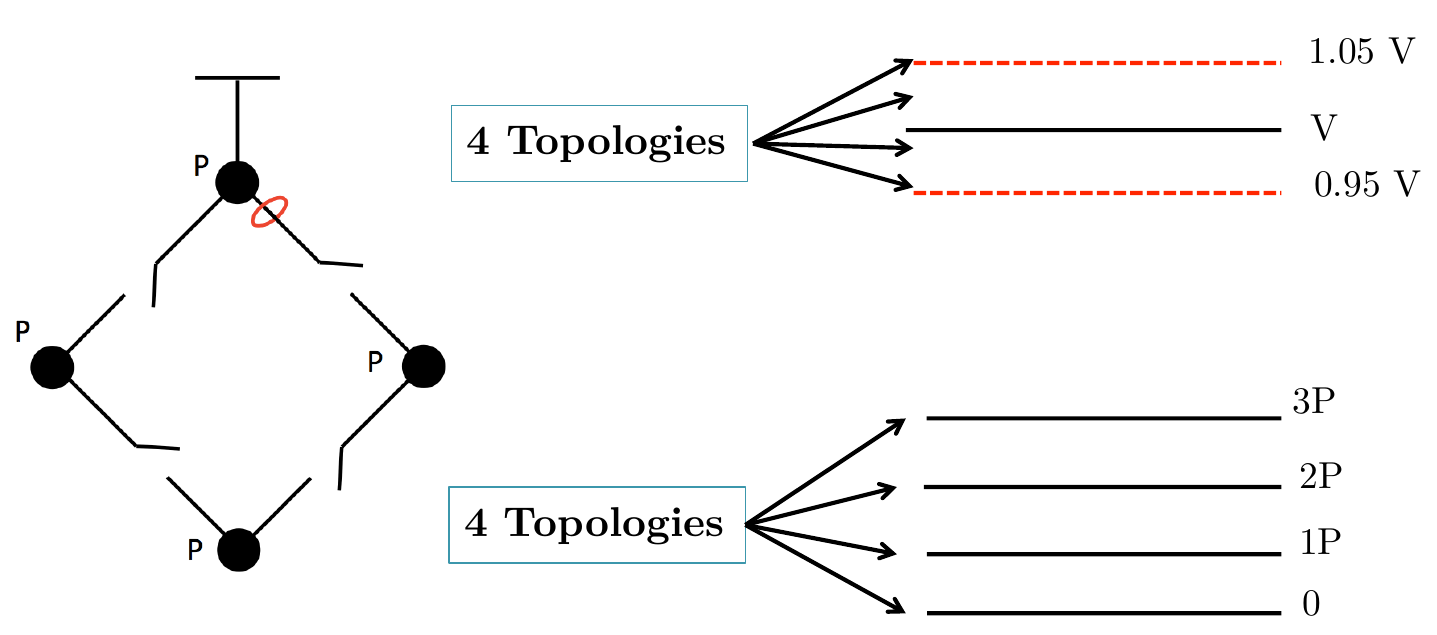}
\label{fig:voltage_vs_power_flow_hypothesis_model}
}

\caption[IEEE Tests and Power Flow Performances]{  
\subref{fig:IEEE123_placement} Performance of max and mean missed detection error for all valid placements $\mbb{M}_{\tau}$.
\subref{fig:IEEE123_voltage_vs_power_flow} Performance of voltage magnitude/phase vs. power flow under ac power flow model.
\subref{fig:voltage_vs_power_flow_hypothesis_model} Conceptual model of voltage vs. power flow performance difference.
}
\end{figure}  

First, we aim to evaluate the missed detection error over the entire set of valid placements, $\mbb{M}_{\tau}$, to characterize their performance.
From Theorem \ref{thm:spanning_tree_identifiability}, the minimum number of flow measurements is $\mu = 4$ and the set $\mbb{M}_{\tau}$ can be generated easily from  \eqref{eq:generate-valid-placement-restricted}.

To simulate the load forecasting error, we rely on the results from \cite{Sevlian2014} to model load day ahead load forecast uncertainty. 
The forecast errors are used to construct the following scaling law for the coefficient of variation: 
\begin{align}
\frac{\sigma}{\mu} =  \sqrt{ \frac{3562}{W} + 41.9} \label{eq:agg-err-curve-varmax1}.
\end{align}. 
Since the loads of each individual island is quite large, the CV of each island is close to $6.3 \%$.
Figure \ref{fig:IEEE123_placement} shows the performance of each sensor placement on the mean and max missed detection error.
We evaluate the set of restricted placements $\mbb{M}_{\tau}$ and spanning trees $\mbb{T}_{\tau}$, where $|\mbb{M}_{\tau}| = 44$.
Notice that for almost half of the placements the maximum error is negligibly small.

This analysis is further explored using a single phase of the 123 test feeders AC power flow model where voltage magnitude, phasor and power flow sensors are compared.
In this test, AMI loads at each node are given with line sensors at switches $w_1$, $w_5$, $w_7$ and $w_9$. 
The loads at each node are the default value of complex power injection.

For each hypothesis $\mc{T}$ the following is computed:
\begin{enumerate}
\item solve the AC power flow $\{ \mb{v}, \mb{s} \} = F( \mb{p}, \mb{q}; \mc{T} )$;
\item generate the three types of measurements for each switch location: 
\begin{align}
z^{true}_{v-mag}(\mc{T}) &= |\mb{v}| \\
z^{true}_{v-phsr}(\mc{T}) &= \mb{v}  \\
z^{true}_{pf}(\mc{T})        &= \mb{s};  
\end{align}
\item add additive noise to each type of measurement with some SNR: $z^{obs}(\mc{T}) = z^{true}(\mc{T}) + \epsilon$;
\item calculate the detector output 

$~~~~~~~~~~\hat{\mc{T}} = \underset{\mc{T} \in \mbb{T}_{\tau}}{\arg\min} \| z^{obs}(\mc{T}) - z^{true}(\mc{T}) \| $.
\end{enumerate}

The results of this experiment are shown in Figure \ref{fig:IEEE123_voltage_vs_power_flow}.
The simulations illustrate that power flow is vastly more powerful in separating hypothesis than voltage magnitude.
For almost all SNR values, power flow measurements are capable of distinguishing all potential hypotheses.
On the other hand, voltage magnitude and phase fail once the $SNR > -10$ dB.

This example allows us to comment on the common understanding in generalized state estimation performance.
This understanding can be seen in Figure \ref{fig:voltage_vs_power_flow_hypothesis_model} where, for a given network of $4$ topologies, the observed values (red flow monitor) are separated very differently in the case of voltage measurements and in power flow measurements. 
In practice the set of all topologies generally map to a range of $\pm 5\%$.  
Therefore, when factoring in uncertainties, the missed detection rates can be high.
On the other hand, measuring flows lead to very large changes in the observation vector over the range of hypothesis.
Power flow measurements separate each hypothesis into a larger space than do voltage measurements.

\begin{figure}[h]
\centering
\subfigure[][]{   
\includegraphics[scale=0.23]{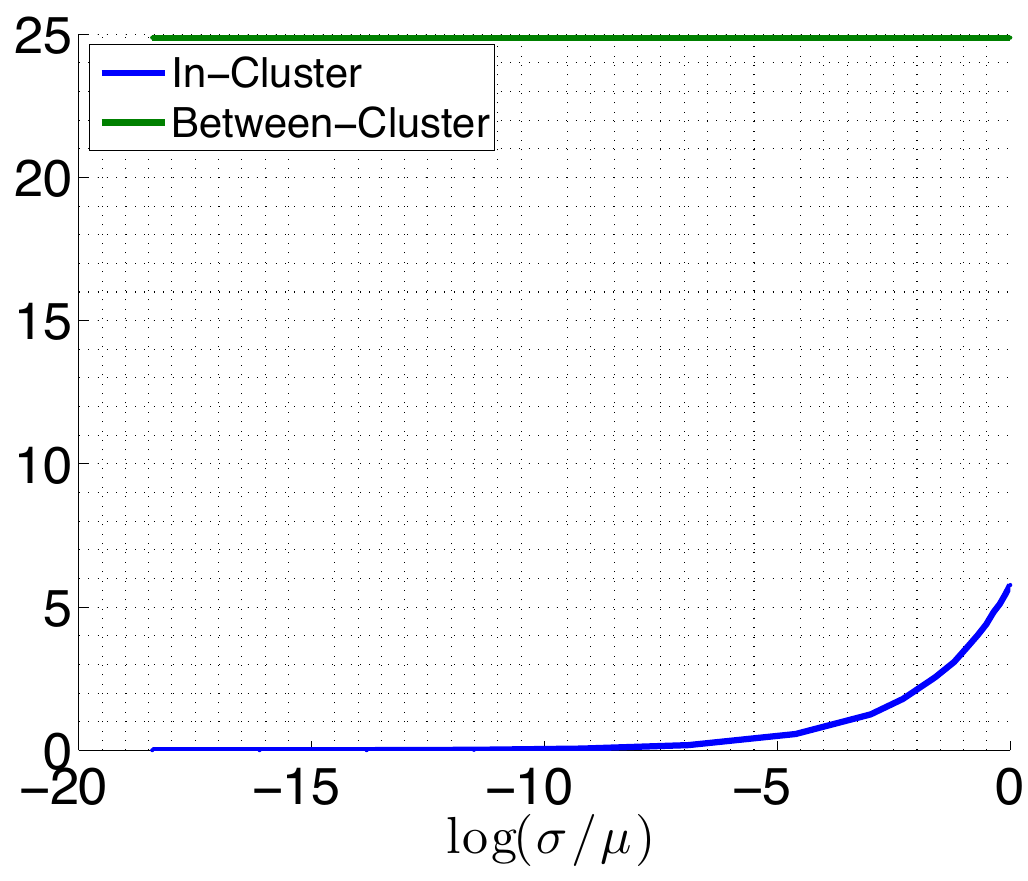}
\label{fig:pf_cluster_analysis}
}
\subfigure[][]{   
\includegraphics[scale=0.23]{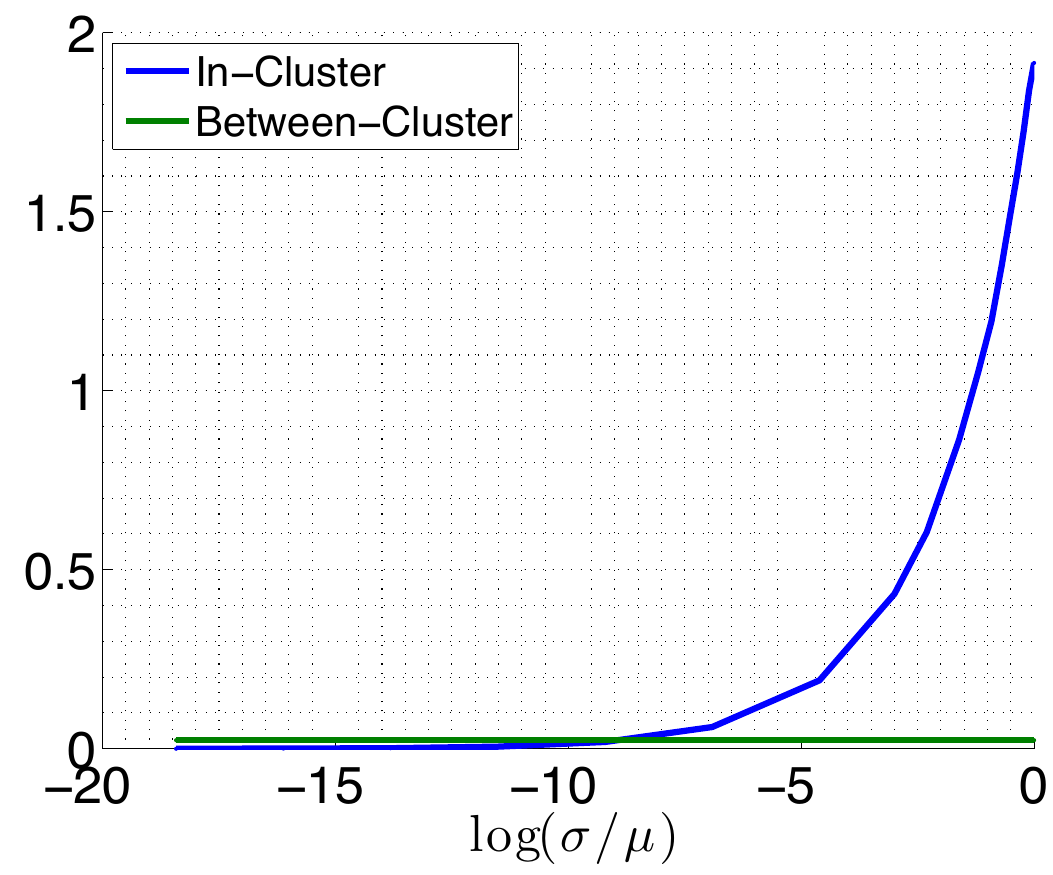}
\label{fig:vmag_cluster_analysis}
}
  
\caption[In cluster and between cluster separation.]{  
In-Cluster-Separation and Between-Cluster-Separation for both power flow  \subref{fig:pf_cluster_analysis} and voltage phasor \subref{fig:vmag_cluster_analysis}.
}
\end{figure}  

This intuition can be quantified by analyzing each measurement point in $\mbb{R}^{4}$ as belonging to some cluster center.
The metrics of between cluster separation and in cluster separation can evaluate this separability as follows:
\begin{itemize}
\item Between-Cluster-Separation (BCS) which quantifies the average separation between all cluster centers.
\begin{align}
\text{BCS} &= \frac{1}{|\mbb{T}|^2} \sum_{\mc{T}, \mc{T}^{\prime} \in \mbb{T} } \| z^{true}(\mc{T}) - z^{true}(\mc{T}^{\prime})  \|;
\end{align}
.
\item In-Cluster-Separation (ICS) which quantifies the variation of observed measurement with respect to the cluster mean.
\begin{align}
\text{ICS} &=  \frac{1}{M |\mbb{T}|} \sum_{\mc{T} \in \mbb{T} } \sum_{m =1 }^{M} \| z^{true}(\mc{T}) - z^{obs}(\mc{T}^{\prime}, m)  \|.
\end{align}
\end{itemize}

The cluster analysis for voltage magnitude and power flow are shown in Figure \ref{fig:pf_cluster_analysis}, \ref{fig:vmag_cluster_analysis}.
The between cluster separation is invariant on SNR and describes how much each observation is separated in the space.
The BCS of the power flow measurements is $25 p.u.$ units, while the voltage measurements are negligible ($0.0087$ p.u.).
Notice these values are fixed and only depend on $z^{true}(\mc{T})$.
Similarly, the ICS increases as the SNR of the measurements increase.
The two measurement types have similar growth, and differ only slightly.
Therefore it is clear to see that the large errors seen using voltage only measurements occur because for any realistic noise value, the ICS $\gg$ BCS for voltage measurements while ICS $\ll$ BCS for power flow measurements, thus verifying experimentally the claim that power flow is more robust for topology detection.

\section{Conclusion}
\label{conclusions}
This paper investigates the problem of topology detection in distribution system using smart meter forecasts and line sensing.
The problem is formulated as a spanning tree detection problem on an `island-graph' and solved for a deterministic and stochastic case.
In the deterministic case, we can guarantee correctness and efficiency of our method, while in the stochastic case we present a combinatorial complexity maximum likelihood detector as well as two approximate algorithms.
Finally, numerical simulations are performed showing the performance of the various methods and detector performance in {\color{black} the IEEE 123 Test feeder}.

\section{Appendix}

\subsection{Nomenclature Table}

\begin{table}[!htb]
\centering     
\caption{Nomenclature}   
\label{tab:FC_construction_from_T_Tprime}   
\begin{tabular}{@{}ll@{}}   
$G(V, E)$  		  		 & undirected graph $G$; vertices $V$; edges $E$ \\
$\mc{T}$   		   		& spanning tree on $G$ \\ 
$\mbb{T}$ 		   		& set of all spanning trees that are constructed on $G$ \\
$\tau$        		   		& set of edges in constructing island graph \\
$\mc{M}$    		   		& sensor placement $(\mc{M} \subset E)$ \\
$\mbb{M}$    	  	   		& set of all sensor placements leading to identifiably in $\mbb{T}$ \\
$x_n, \hat{x}_n$	   		& true and forecasted load of node $v_n$ \\
$\sigma^2_n, \Sigma$ 	   	& forecast error variance and covariance matrix \\
 $\Gamma(\mc{T}, \mc{M})$      & observation matrix for tree $\mc{T}$, sensor placement $\mc{M}$ \\
$\mb{s}_{obs}$                          & set of measured power flow \\
$\mb{s}(\mc{T}, \mb{x})$   	& true and predicted flow measured under hypothesis $\mc{T}$ \\
$g_{1,2}(\mc{M})$                      & (1) maximum / (2) mean missed detection \\
$c$                                 		& cycle in graph $G$ \\ 
$\mc{C}(G)$                   		& cycle space of $G$ or set of all possible cycles\\
$\mc{FC}(\mc{T})$          		& fundamental cycle Basis of $G$ constructed by $\mc{T}$ \\
$\mc{FC}_{\mc{M}}$       		& fundamental cycle Basis constructed by $\mc{M}$ \\
$\lambda_k$  				& $\lambda_k$ is $k^{\text{th}}$ cycle in $\mc{FC}_{\mc{M}}$ \\
$\mu(G)$                         		& circuit rank of graph \\
$n(G)$                            		& Number of connected components \\
$\Delta E$                       		& edge exchange operation to generate new tree: $\mc{T} \rightarrow \mc{T}^{\prime}$ \\
$\mc{K}(c)$                     		& cycle-sensor map indicating all sensors on cycle $c$ \\ 
\end{tabular}
\end{table}
\subsection{Useful Graph Theory Definitions and Results}
\label{section-Useful-Graph-Theory-Definitions-and-Results}  

\begin{figure}[h]
\centering
\subfigure[][]{    
\includegraphics[scale=0.47]{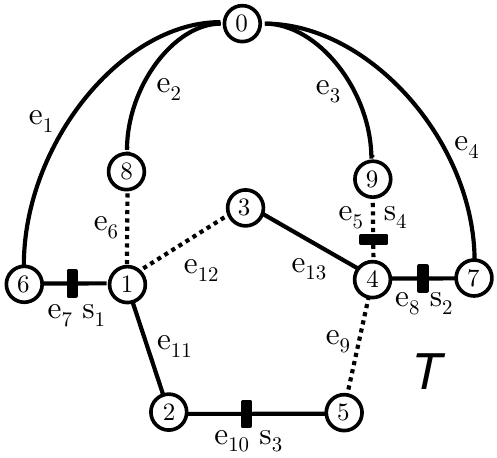}   	
\label{fig:appendix_T1}
}
\subfigure[][]{ 
\includegraphics[scale=0.47]{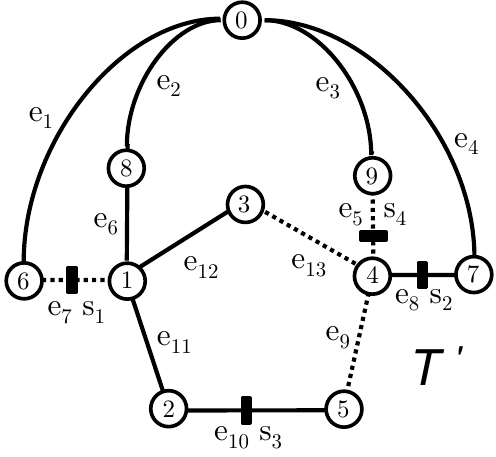}
\label{fig:appendix_T2}
}
\caption[Two spanning trees to illustrate properties.]{
Two spanning trees $\mc{T}$ (\ref{fig:appendix_T1}) and $\mc{T}^{\prime}$ (\ref{fig:appendix_T2}), with sensor placement $\mc{M}$.
}
\end{figure}

\noindent
Refer to \cite{Diestel2000} for a more thorough presentation.  

\noindent
\textit{Cycle}:
A cycle $c = \{e_{1}, \hdots, e_{N} \}$ is a connected subgraph where each vertex is of degree 2.

\smallskip
\noindent
\textit{Cycle Space}:
The set $\mc{B}_E$ is the power set over the edge set $\mc{B}_E = \{0, 1\}^{|E|}$.
Any cycle $c$ is a vector in the space $\mc{B}_E$.
Vector addition is defined as $c^{\prime} = c_1 \oplus c_2$ where new cycles are constructed via symmetric difference operation on edges: $E_1 \bigoplus E_2 = (E_1 \cup E_2) \setminus (E_1 \cap E_2)$.
The cycle space $\mc{C}(G)$ of the graph is the vector space of all possible cycles in a particular graph. 

For example consider the cycles in the graph $G$ in Figure \ref{fig:appendix_T1}.
For each of the dashed edges ($e_6$ and $e_7$), adding them back to the tree will construct cycles associate with them, along with cycle formed by their addition:
\begin{align*}
c_1 &= \{e_7, e_1, e_2, \underline{e_6} \} \\
c_2 &= \{e_1, e_2, e_3, e_5, e_{13}, \underline{e_{12}} \} \\
c_3 &= c_1 \oplus c_2 =  \{e_7, e_1, e_3, e_4, e_{13}, \underline{e_{12}} \}.
\end{align*}
It is easy to see that $c_1$, $c_2$ and $c_3$ are all cycles in $G$.    

\smallskip
\noindent
\textit{Circuit Rank}:
The circuit rank of a graph is given by $\mu = |E| - |V| + n(G)$ where $n(G)$ is the number of connected components of the graph.
For example, the island graph in Figure \ref{fig:appendix_T1} has $n(G)=1$, and $\mu = |13| - |10| + 1 = 4$.

\smallskip
\noindent
\textit{Cycle Basis}: 
The analog of a vector basis for cycle spaces is the cycle basis.
A basis $\mc{B}_{C} \subset \mc{B}_{E}$ is the smallest number of cycles whereby all other cycles can be constructed via symmetric difference operations.
The dimension of $\mc{C}(G)$ and $\mc{B}_{C}$ is $\mu$, the circuit rank of the graph.
Therefore, $\mu(G)$ is the smallest number of cycles required to produce all other cycles on a graph.
We use $\mu(G)$ and $\mu$ interchangeably, whenever convenient.

For example, all cycles in the graph in Figure \ref{fig:appendix_T1} can be constructed from $\mu$ independent cycles which form the cycle basis of the graph.

\smallskip
\noindent
\textit{Fundamental Cycle Basis}:
A Fundamental Cycle Basis $\mc{FC}$ is a cycle basis constructed using the following procedure: given a spanning tree $\mc{T}$, enumerate the set of edges in $G$ but not $\mc{T}$.
Then for each $e \in E \setminus \mc{T}$, construct $\mc{T} + e$ then find the single cycle $c$ associated with $e$.
So we can generate $\mu$ cycles in $\mc{FC}$, which is the dimensionality of the basis.
An equivalent definition for a Fundamental Cycle Basis is that each cycle will have one unique edge which is in no other cycle.

\begin{table}[!htb]
\centering     
\caption{Fundamental Cycle Basis from $\mc{T}$, $\mc{T}^{\prime}$ and cycle-measurement map $\mc{K}(c)$}   
\label{tab:FC_construction_from_T_Tprime}   
\begin{tabular}{@{}lll@{}}   
\toprule        
 $\mc{FC}(\mc{T})$                                                                   & &  $\mc{K}(c)$  \\
 $c_1 = \{e_{6}, e_{2}, e_{1}, e_{7} \} $                                     & & $\mc{K}(c_1) = \{s_1\}$   	            \\  
 $c_2 = \{e_{12}, e_{13}, e_{8}, e_{4}, e_{1} , e_{7}\}$             & &  $\mc{K}(c_2) = \{ s_1, s_2 \}$           \\   
 $c_3 = \{e_{5}, e_{3}, e_{4}, e_{8} \}$                                      &  &   $\mc{K}(c_3) = \{s_2, s_4\}$           \\   
 $c_4 = \{e_{9}, e_{8}, e_{4}, e_{1}, e_{7}, e_{11}, e_{10} \}$   & &$\mc{K}(c_4) = \{s_1, s_2, s_4\}$        \\   
                                                                                                  &  &                                                           \\
 $\mc{FC}(\mc{T}^{\prime})$                                                     & &   $\mc{K}(c)$                                       \\
 $c_1 = \{e_{7}, e_{6}, e_{2}, e_{1} \} $      				      & &   $\mc{K}(c_1) = \{s_1\}$                     \\   
 $c_2 = \{e_{13}, e_{8}, e_{4}, e_{2}, e_6, e_{12} \} $               & &     $\mc{K}(c_2) = \{ s_1, s_2 \}$         \\   
 $c_3 = \{e_{5}, e_{8}, e_{4}, e_{3} \} $                                     &  &  $\mc{K}(c_3) = \{s_2, s_4\}$             \\     
 $c_4 = \{ e_{9}, e_{8}, e_{4}, e_{2}, e_{6}, e_{11},  e_{10} \}$ &  &    $\mc{K}(c_4) = \{s_1, s_2, s_4\}$   \\
\bottomrule 
\end{tabular}
\end{table}
From spanning trees $\mc{T}$ and $\mc{T}^{\prime}$ in Figure \ref{fig:appendix_T2}, \ref{fig:appendix_T2} we construct the following Fundamental Cycle Basis in Table \ref{tab:FC_construction_from_T_Tprime} Column 1.
Note that $\mc{FC}(\mc{T}) \neq \mc{FC}(\mc{T}^{\prime})$.
However, this does not always occur, see \cite{Syslo1982} for more details.

\subsection{Proof of Theorem \ref{thm:spanning_tree_identifiability}}
Recall, the theorem states that as long as $G \setminus \mc{M}$ forms a spanning tree, then for any two trees $\mc{T} \neq \mc{T}^{\prime}$, $\mb{s}(\mc{T}, \mb{x}) \neq \mb{s}(\mc{T}^{\prime}, \mb{x})$.
Equivalently $\Delta \mb{s} \equiv \mb{s}(\mc{T}, \mb{x}) - \mb{s}(\mc{T}^{\prime}, \mb{x}) \neq 0$.

We now prove Theorem \ref{thm:spanning_tree_identifiability} in the following steps.
\begin{itemize}    
\item [1] 
We Introduce a cycle-measurement mapping object $\mc{K}(c)$ which tracks sensors on a cycle and show that an independence property if  $G \setminus \mc{M}$ forms a spanning tree.
\item [2] 
We construct an edge exchange procedure which encodes the transition: $\mc{T} \rightarrow \mc{T}^{\prime}$ between any two spanning trees as a set of single cycle edge exchanges on the cycles of $\mc{FC}(\mc{T})$.  
We show that this encoding always exists.
\item [3] 
We show that sensor measurements in $\mc{K}(c)$ decouple from one cycle to another under single edge exchanges.
\item [4]
We use the independence of $\mc{K}(c)$ and decoupling of single edge exchange measurements in $c$ to show inductively that no multiple edge exchanges of any size exist where $\Delta \mb{s} = 0$, if $G \setminus \mc{M}$ is a spanning tree.
\end{itemize}  
\subsubsection{ Cycle-Measurement Map}  
\label{subsubsection-cycle-measurement-map}

The cycle-measurement map encodes which measurements lie on the edges of a particular cycle.
The map $\mc{K}(c)$ is defined according to cycles $c \in \mc{FC}$, for any arbitrary $\mc{FC}$ in G.

\begin{define}
With respect to some $\mc{FC}$, a cycle-measurement map is $\mc{K}:c\rightarrow\mc{M}$ for all $c \in \mc{FC}$ where $s_k \in \mc{K}(c)$ if $s_k$ is on an edge in $c$.
\end{define}
We can also write it as $\mc{K}(c) = c \cap \mc{M}$, though this is an abuse of notation.

For the fundamental cycles associated with the tree in Figure \ref{fig:appendix_T1}, \ref{fig:appendix_T2} we construct the following map shown in Table 
\ref{tab:FC_construction_from_T_Tprime} Column 2.
We now aim to develop some useful properties of this mapping function.

Consider placement $\mc{M}$, and constructed tree $\mc{T} = G \setminus \mc{M}$.
We must have $\mc{K}(c_k) = \{s_k\}$ where $c_k$ is the $k^{\text{th}}$ cycle in $\mc{FC}(\mc{T})$.  
This is obvious by construction.

We denote $\lambda_k$ to be the $k^{\text{th}}$ cycle in $\mc{FC}(G \setminus \mc{M})$ ($\mc{FC}_{\mc{M}}$ shorthand).
It is clear that $\mc{K}(\lambda_k) = \{s_k\}$, for these cycles.

We see that by looking at $\mc{K}(c)$, for any cycle in an arbitrary $\mc{FC}$, we can encode the cycles construction using basis 
$\mc{FC}_{\mc{M}} = $$\{\lambda_1, \hdots, \lambda_{\mu} \}$.
\begin{lem}
If $G \setminus \mc{M}$ forms a spanning tree, then for any $\mc{FC}$ and $c \in \mc{FC}$, $c = \underset{k:s_k \in \mc{K}(c) }{\bigoplus} \lambda_k$.
\end{lem}
\begin{proof}
Any cycle $c$ can be represented as a combination of cycles $\lambda \in \mc{FC}_{\mc{M}}$ since the $\mc{FC}_{\mc{M}}$ is a cycle basis.
If any other $\lambda^{\prime}$ outside of the set $\{\lambda_k : s_k \in \mc{K}(c) \}$ is used to construct $c$ then $\mc{K}(c)$ will include edge containing a $s^{\prime}$. 
Conversely if any additional $\lambda^{\prime}$ is needed to construct $c$, it's $s$ must be in $\mc{K}(c)$.
\end{proof}
Now we can prove a general case of `independence' between any two subsets of cycles and the measurements that are placed on them.
\begin{lem}
\label{cycle-sensor-linear-independence}
If $G \setminus \mc{M}$ forms a spanning tree, then for any $\mc{FC}$ and subsets $A \neq B$ of cycles in $\mc{FC}$ we must have that: 
\begin{align}
\underbrace{  \underset{k \in A}{\bigcup} \mc{K}(c_k) }_{\mc{K}_A} \neq \underbrace{ \underset{k \in B}{\bigcup} \mc{K}(c_k) }_{\mc{K}_B}. 
\end{align}
\end{lem}     

\begin{proof}
Suppose that there exists some $\mc{FC}$ and partitions $A$, $B$ where the terms $\mc{K}_A$ and  $\mc{K}_B$ are equal.
Since $\mc{K}_A$ and $\mc{K}_B$ encode some cycle we have that $c_A = c_B$.
However, since $c_A$ and $c_B$ {\color{black}are by} definition fundamental cycles in $\mc{FC}$ there will exist an edge in $c_A$ that is not in $c_B$, thus $c_A \neq c_B$.
\end{proof}

To see an example of this, consider the cycle-measurement map generated by $\mc{FC}(\mc{T})$ in Figure \ref{fig:appendix_T1} in Table \ref{tab:FC_construction_from_T_Tprime}.
Notice that we cannot construct any partitions A, B where all the covered measurements are equal.
For example if $A = \{c_1, c_2 \}$ and $B = \{c_3, c_4 \}$, we have that $\underset{k \in A}{\bigcup} \mc{K}(c_k) = \{s_1, s_2 \}$ and $\underset{k \in B}{\bigcup} \mc{K}(c_k) = \{s_1, s_3, s_4 \}$.

Notice that the cycle-measurement-map $\mc{K}$ constructed for the network in Figure \ref{fig:island_graph_with_spanning_tree} does not satisfy that $G \setminus \mc{M}$ is a spanning tree.
In this case we have that $\mc{K}(c_1) = \{s_1 \}$, $\mc{K}(c_2) = \{s_1, s_2 \}$, $\mc{K}(c_3) = \{s_2 \}$ and $\mc{K}(c_4) = \{s_1, s_2 \}$ for the same cycles.
The cycle partition $A = \{c_1\}$, $B= \{c_4 \}$ clearly leads to the independence property not holding.
  
This result leads to the following equivalent results which are used in our proof.
\begin{cor}
\label{cor-unique-sensor-in-each}
A special case is that $\forall c, A \subset \mc{FC}$, $\exists s \in \mc{K}(c)$ s.t. $s \not\in \underset{k \in A}{\bigcup} \mc{K}(c)$, for any $A$ not including $c$.
\end{cor}

\begin{RK}    
\label{cor-subset-unique-sensor-count}
The subset $C \subset \mc{FC}$, where $|C| = N$ will have at least $N$ unique sensors.
\end{RK}

\begin{RK}  
The indicator vector associated with each $\mc{K}(c)$ are a set of linearly independent vectors.
\end{RK}

\subsubsection{ Edge Exchange Operator}
We encode the transition between any any two spanning trees and show that any such transition can be represented as a set of single edge changes.
To motivate this, consider the trees $\mc{T}$ and $\mc{T}^{\prime}$ in Figure \ref{fig:appendix_T1} and \ref{fig:appendix_T2}.
The removed edges from $G$ in each tree are $E_{R} = \{e_6, e_{12}, e_9, e_5 \}$ and $E^{\prime}_{R} = \{e_7, e_{13}, e_9, e_5 \}$.
In both trees, $e_9$ and $e_5$ do not change.

The main question we want to answer is how to encode the transition between trees by single edge exchanges. 
Namely, if we define a $\Delta E$ operation, do we encode as $\Delta E = \{e_6 \rightarrow e_7, e_{12} \rightarrow e_{13} \}$ or $\Delta E = \{e_6 \rightarrow e_{13}, e_{12} \rightarrow e_{7} \}$. 

\begin{table}[!htb]
\centering     
\caption{Fundamental Cycle Basis from $\mc{T}$, $\mc{T}^{\prime}$ and cycle-measurement map $\mc{K}(c)$}   
\label{tab:T_Tprime_edge_exchange}   
\begin{tabular}{@{}lllll@{}}   
\toprule          
 $\mc{FC}(\mc{T})$                                                                     & &  $c \cap E_R$      & &         $c \cap E_R$                         \\
 $c_1 = \{e_{6}, e_{2}, e_{1}, e_{7} \} $                                      & &   $\{e_{6} \}$       & &              $\{ e_7 \}$    \\  
 $c_2 = \{e_{12}, e_{13}, e_{8}, e_{4}, e_{1} , e_{7}\}$              & &   $\{e_{12} \}$     & &              $\{ e_7, e_{13} \}$                 \\   
 $c_3 = \{e_{5}, e_{3}, e_{4}, e_{8} \}$                                       & &   $\{e_9 \}$         & &              $\{ e_9 \}$     \\   
 $c_4 = \{e_{9}, e_{8}, e_{4}, e_{1}, e_{7}, e_{11}, e_{10} \}$    & &   $\{ e_{5} \}$     & &               $\{ e_5 \}$      \\   
\bottomrule 
\end{tabular}
\end{table}

This can be resolved if we look at edge exchanges with respect to $\mc{FC}(\mc{T})$, as shown in Table \ref{tab:T_Tprime_edge_exchange}.
Column 1 repeats the cycles in $\mc{FC}(\mc{T})$.
Column 2 maps edges in $E_{R}$ onto $\mc{FC}(\mc{T})$ and Column 3 maps edges in $E^{\prime}_{R}$ onto $\mc{FC}(\mc{T})$.
This can be seen as an identical mapping function as $\mc{K}(c)$ in Section \ref{subsubsection-cycle-measurement-map}, except we replace edges with measurements on them, with edges in $E_{R}$.
Notice that we can now define a cycle by cycle set of edge exchanges that define the transition from $\mc{T} \rightarrow \mc{T}^{\prime}$.
So on cycle $c_1$ we have $\Delta e  =  (e_6 \rightarrow e_7)$ and on cycle $c_2$ we have $\Delta e  =  (e_{12} \rightarrow e_{13})$. 

\begin{define}
An edge exchange with respect to $\mc{FC}$ is  $\Delta E = \{ \Delta e_1, \hdots, \Delta e_\mu \}$ where $\Delta e_k = (e_k \rightarrow e^{\prime}_k)$, $e_k \in E_R$,  $e^{\prime}_k \in E^{\prime}_R$ and $e_k, e^{\prime}_k  \in c_k$.
\end{define}

We can generate an edge exchange encoding as follows.
First assign $e_k$ the edge in $E_R$ used to construct $c_k \in \mc{FC}(\mc{T})$.  
For $e^{\prime}_k$, we use the following procedure:
\begin{enumerate}
\item starting at $c_1$, set $e^{\prime}_1$ to be any element in $c \cap E^{\prime}_R$;
\item for the $k^{th}$ $c_k$, set $e^{\prime}_k$ to be any element in $c_k \cap E^{\prime}_R$ that has not already been assigned to previous $e^{\prime}_1 \hdots, e^{\prime}_{k-1}$.
\end{enumerate}

\begin{lem}
\label{lem-edge-exchange-existence}
For any two $\mc{T}, \mc{T}^{\prime}$, at least one edge exchange procedure always exists.
\end{lem}
\begin{proof}
The set of edges $c \cap E_R$ are equivalent to the mapping $\mc{K}(c)$.
Therefore, Corollary \ref{cor-unique-sensor-in-each}, holds for each incremental cycle to be processed: that is, every additional $c \cap E_R$ set will have a one new edge that can be assigned to $e^{\prime}_k$.
\end{proof}

Therefore, between any two spanning trees there is a well defined set of single edge exchanges performed on the cycles of $\mc{FC}(T)$ which encode any arbitrary $\mc{T}\rightarrow \mc{T}^{\prime}$.
  

\begin{figure}[h]
\centering
\subfigure[][]{ 
	\includegraphics[scale=0.4]{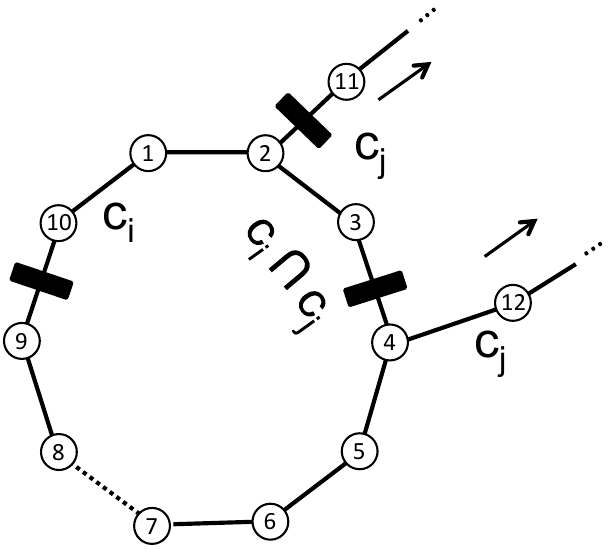}
	\label{fig:2cycle_crop}
}
\caption[Illustrative 2 cycle example.]{
	\ref{fig:2cycle_crop}
	Illustrative example to see properties in Lemma \ref{lem:decouple-on-fc}.
	}
\label{fig:fc_edge_exchange}
\end{figure} 

\subsubsection{ Decoupling of Measurement along Cycle}
We show that when an edge exchange occurs on a cycle, we need to only consider changes of flow values in $\mc{K}(c)$.
In the development $\mc{K}(c)$ and edge exchanges are focused on cycles in $\mc{FC}$.

\begin{prop}
\label{prop:exchange-path-rearrange}
A single edge exchange on $c$, with vertices $\{v_0, \hdots, v_m\}$, results in a permutation of an uninterrupted path of the vertices.
Therefore, if $e \neq e^{\prime}$, $p(c, e) \neq p(c, e^{\prime})$.
\end{prop}

\begin{lem}
\label{lem:decouple-on-fc}
Consider a single edge exchange on $c$, $e \rightarrow e^{\prime}$, where the following holds:
\begin{itemize}
\item [P1] $\forall s \not\in \mc{K}(c)$, $\Delta \mb{s} = 0$, 
\item [P2] $\forall s \in \mc{K}(c)$, $\Delta \mb{s} \neq 0$.      
\end{itemize}
\end{lem}

\begin{proof}
(1) From Proposition \ref{prop:exchange-path-rearrange}, any sensor that measures a single vertex in $p$ will measure all the vertices in $p$ before and after the edge exchange.
(2) From Proposition \ref{prop:exchange-path-rearrange}, since we rearrange all the nodes yet keep a fixed edge to measure flows, all sensors in $\mc{K}(c)$ will change values.
\end{proof}

\begin{RK}
\label{remark-single-independent-sensor}
In condition (2) of Lemma \ref{lem:decouple-on-fc}, we have $\Delta \mb{s} \neq 0$ since we measure both magnitude and direction of flow.
\end{RK}

\subsubsection{Inductive Proof of Theorem \ref{thm:spanning_tree_identifiability}}
\label{subsubsection-inductive-proof}
  
Note our effort is to show sufficiency, for necessity we only need to consider the cycle $c$ in Figure \ref{fig:2cycle_crop}.
If no sensor exists on $c_i$, $E \setminus \mc{M}$ is not a spanning tree and every edge exchange leads to $\Delta \mb{s} = 0$.

From Lemma \ref{lem-edge-exchange-existence}, we can encode any $\mc{T} \rightarrow \mc{T}^{\prime}$ transition as a set of cycle-based edge exchanges.
\textit{We now show inductively that no edge exchange exists, which leads to $\Delta \mb{s} = 0$.}
So $\mc{T} \rightarrow \mc{T}^{\prime}$ always leads to $\Delta \mb{s} \neq 0$.

We show that given any number of non-trivial edge exchanges, there is at least one $s_i$ s.t. $\Delta \mb{s}_i \neq 0$. 
\begin{itemize}
\item [] \textbf{Base Case}
	Assume $c \subset \mc{FC}$ contain a non-trivial edge exchange.
        From Lemma \ref{lem:decouple-on-fc} (P2), $\forall s \in \mc{K}$, $\Delta \mb{s} \neq 0$.
\item [] \textbf{Inductive Hypothesis} 
	Assume multiple cycles $C \subset \mc{FC}$, contain a non-trivial edge exchange, where $|C|$ = $N$.
	Assume there exists at least one $s_i \in \mc{K}_{C}(c)$, where $\Delta \mb{s}_i \neq 0$.
\item [] \textbf{Inductive Step}
	Suppose we find a cycle $c_{n+1} \in \mc{FC} \setminus C$ where an edge exchange leads to $\Delta \mb{s}_i = 0$.
	From Corollary \ref{cycle-sensor-linear-independence}, subset, $C \cup c_{n+1}$ will have at least $N+1$ unique sensors and the new cycle must introduce some sensor $s_j \not\in \mc{K}_{C}$.
	From Lemma \ref{lem:decouple-on-fc} (P2) $\Delta \mb{s}_j \neq 0$.
\end{itemize}

\subsection{ Proof of Theorem \ref{thm:undirected-flow-placement} }

This can first be shown in the following example.
In Figure \ref{fig:2cycle_crop}, where the cycle in consideration $c_i$ will have a single edge exchange: for example $\Delta e = (8,7) \rightarrow (1, 2)$ where all $x_i = 1$.
In this case, $s = 2$ and $s^{\prime} = -2$.
Since we assume that the magnitude and direction of the flow is measured, therefore there is no ambiguity.
Any other edge will again lead to a new measured value.
Since nodes $\{v_0, \hdots, v_{10}\}$ are always connected, all sensors outside of $\mc{K}(c)$, for example a measurement on edge $e = (2, 10)$ will not detect an edge exchange.

\begin{proof}
If we only consider single edge exchange, ambiguity occurs if a cycle has a single measurement as in the example.
If $|p_i| = 2$, no ambitious single edge exchange can occur.
If $|p_i| > 2$, a single edge exchange will have no ambiguity if there are at least $2$ measurements on any cycles where a single exchange occurs.
\end{proof}
    
\subsection{ Proof of Lemma \ref{lem-BN-rank-N-1} }
    
\begin{proof}
The matrix $B \in \{-1, 0, +1\}^{N \times M}$ and a valid placement of size $\mu = M-N+1$, therefore, $B_N \in \{-1, 0, +1\}^{N \times N-1}$.
Since the edges of $G \setminus \mc{M}$ maintain a spanning tree property, the graph has 1 connected component. 
Therefore the incidence matrix must be of rank $N-1$ \cite{GraphsAndMatrices2010}.
Since matrix $B^{r}_N$ is a square matrix of size $N-1$ and rank $N-1$ it is invertible.  
\end{proof}

\subsection{Proof of Theorem \ref{thm-f-star-correct-spanning-tree}}
\label{subsection-proof-of-Theorem-2}

\begin{proof}
We need to show that  \eqref{f-star-solution} encodes the correct spanning tree: so $\mb{f}^{\star} = \mb{f}(\mb{x}, {\mc{T}_1})$.
This can be done by contradiction: 
Assume that the solution to $\mb{f}^{\star} = \mb{f}(\mb{x}, {\mc{T}_2})$ represents some other spanning tree or even connected graph with some flow.
This implies that
\begin{align}
\begin{bmatrix} 
(B^{r,-1}_N) (\mb{x} - B^{r}_{M} \mb{s}({\mc{T}_2}, \mb{x}) )  \\ \mb{s}({\mc{T}_2}, \mb{x})
\end{bmatrix}  
= 
\begin{bmatrix} 
(B^{r,-1}_N) (\mb{x} - B^{r}_{M} \mb{s}({\mc{T}_1}, \mb{x}) )   \\ \mb{s}({\mc{T}_2}, \mb{x})
\end{bmatrix}.  
\end{align}
\color{black}
Since $(B^{r,-1}_N)$ exists, this implies that both $\mb{s}({\mc{T}_2}, \mb{x}) = \mb{s}({\mc{T}_1}, \mb{x})$ and $B^{r}_{M} \mb{s}({\mc{T}_2}, \mb{x}) = B^{r}_{M}  \mb{s}({\mc{T}_1}, \mb{x})$.
Since the dimension of the null space of $B^{r}_{M}$ is $0$, this reduces to to $\mb{s}({\mc{T}_1}, \mb{x}) = \mb{s}({\mc{T}_2}, \mb{x})$.
However from Theorem \ref{thm:spanning_tree_identifiability}, this is a contradiction.

Finally, we must show that no subgraph of $G^{\prime}$ is uniquely distinguished from $\mc{T}_1$.
Such a subgraph occurs for any spanning tree, where a removed edge of the co-tree is added to $\mc{T}$ (i.e. $G^{\prime} = \mc{T}+e$, for $e \in G \setminus \mc{T}$).

We show that $\mb{s}({\mc{T}_1}, \mb{x}) \neq \mb{s}(G^{\prime}, \mb{x})$ repeating the a proof similar to Theorem \ref{thm:spanning_tree_identifiability} as follows:
\begin{enumerate}
\item Repeat lemma \ref{lem:decouple-on-fc} for not only edge exchanges, but the case where the cycle has no removed edge.
We can distinguish any added edge vs. spanning tree on the cycle by only measurements on that cycle.
All measurements not on the cycle will not see any change in topology. 
Therefore (P1) and (P2) of lemma \ref{lem:decouple-on-fc} still hold.
\item Given the single cycle property, for distinguishing spanning trees and any $\mc{T}+e$, for a single edge addition, we can repeat the same proof for Theorem \ref{thm:spanning_tree_identifiability} for arbitrary edge additions thereby showing that $\mb{s}({\mc{T}_1}, \mb{x}) \neq \mb{s}(G^{\prime}, \mb{x})$.
\end{enumerate}
\end{proof} 


\subsection{Proof of Theorem \ref{thm-comb-flow-rewrite} }
\label{subsection-proof-of-theorem-flow}

Recall the combinatorial ML detector is the following:
\begin{align}
\mc{T} =  \underset{ \mc{T} \in \mbb{T}_{\tau} }{\arg\max} \left(  \mb{s}_{obs} - \mb{s}(\mb{x}, \mc{T} )  \right)^{T}  \Sigma^{-1}_{s, i}  \left(  \mb{s}_{obs} - \mb{s}(\mb{x}, \mc{T} )  \right). \label{eq:appendix-comb-ML}
\end{align}

Observation vector $\mb{s}_{obs} = \{ \mb{s}_{obs, +}, \mb{s}_{obs, z} \}$ leads to a reduction of the initial search space from $\mbb{T}$ to $\mbb{T}_{+}$ by removing any tree which will violate the zero/non-zero flow observations. 
Additionally, we must remove the zero observations in the likelihood function and reduce the covariance matrix.

Recall that we can construct the observations $\mb{s} = A_{\mc{M}} \mb{f}$, therefore we can construct a reduced $\mb{s}_{+} = A_{\mc{M}, +} \mb{f}$, by removing the rows associated with the zero observations.
We can similarly remove the specific rows and columns of the covariance matrix by a matrix $E$.
The matrix $E$ is constructed by removing the columns of the identity matrix corresponding to the index of the zero observations.
Therefore $\Sigma_{+} = E \Sigma E^{T}$, resulting in the true covariance matrix:
\begin{align*}
\Sigma_{s, \mc{T}, +} = \sigma^2 E  A_{\mc{M}} B^{r, -1}_{\mc{T}} B^{r,-1, T}_{\mc{T}} A^T_{\mc{M}} E^{T}.
\end{align*}
This will guarantee that $\Sigma^{-1}_{s, \mc{T}, +}$ always exists.
We can now re-arrange the combinatorial optimization over $\mc{T}$ and $\mb{s}(\mb{\hat{x}}, \mc{T})$ in terms of a power flow vector $\mb{f}$ leading to our desired reduction.

Starting from \eqref{eq:appendix-comb-ML}, we have the following (shown on following page):

\begin{figure*}
\begin{align}
\hline
\mc{\hat{T}}  &= \underset{\mc{T} \in \mbb{T}^{+}}{\arg\min} \frac{1}{2} \left( \mb{s}_{obs, +} - \mb{s}_{+}(\mb{\hat{x}}, \mc{T}) \right)^{T} \Sigma^{-1}_{s, \mc{T}, +}  
		     \left( \mb{s}_{obs, +} - \mb{s}_{+}(\mb{\hat{x}}, \mc{T}) \right)   	
		     - \frac{1}{2} \ln \left( \det( \Sigma^{-1}_{s, \mc{T}, +} )\right)  \label{s-to-f-eq1}  	\\
	            &= \underset{\mc{T} \in \mbb{T}^{+}, A_{M, +} \mb{f} = \mb{s}_{obs, +}  }{\arg\min} 
	                \frac{1}{2} \left( A_{M, +} \mb{f} - A_{M, +}\mb{f}(\mb{\hat{x}}, \mc{T}) \right)^{T} \Sigma^{-1}_{s, \mc{T}, +} \left( A_{M, +}  \mb{f} - A_{M, +}\mb{f}(\mb{\hat{x}}, \mc{T}) \right)  
	                  - \frac{1}{2} \ln \left( \Sigma^{-1}_{s, \mc{T}, +} )\right)  \label{s-to-f-eq2}  \\	    
	      &= \underset{\mc{T} \in \mbb{T}^{+},~ \mb{f}: A_{M, +} \mb{f} = \mb{s}_{obs, +} }{\arg\min} \frac{1}{2} \left(   \mb{f} -\mb{f}(\mb{\hat{x}}, \mc{T}) \right)^{T} A^T_{M, +} 
	            \Sigma^{-1}_{s, \mc{T}, +}  A_{M, +}  \left( \mb{f} -\mb{f}(\mb{\hat{x}}, \mc{T}) \right) 
	            - \frac{1}{2} \ln \left( \det( \Sigma^{-1}_{s, \mc{T}, +} )\right)    \label{s-to-f-eq3}  \\	    
	     &= \underset{\mc{T} \in \mbb{T}^{+},~ \mb{f}: A_{M, +} \mb{f} = \mb{s}_{obs, +} }{\arg\min}  \frac{1}{2} \left( \mb{f} - \mb{f}(\mb{\hat{x}}, \mc{T}) \right)^{T} B^{r, T}_{\mc{T}} B^{r, T, -1}_{\mc{T}} 
	           A^T_{M, +} \Sigma^{-1}_{s, \mc{T}, +}  A_{M, +}  B^{r}_{\mc{T}} B^{r, -1}_{\mc{T}}  \left( \mb{f} -\mb{f}(\mb{\hat{x}}, \mc{T}) \right)  
	           - \frac{1}{2} \ln \left( \det( \Sigma^{-1}_{s, \mc{T}, +} )\right)  \label{s-to-f-eq4}  \\	    	    	    	    	    
	    &= \underset{\mc{T} \in \mbb{T}^{+},~ \mb{f}: A_{M, +} \mb{f} = \mb{s}_{obs, +} }{\arg\min} \frac{1}{2} \left(  B^{r}_{\mc{T}} \mb{f} - B^{r}_{\mc{T}} \mb{f}(\mb{\hat{x}}, \mc{T}) \right)^{T} B^{r, T, -1}_{\mc{T}} A^T_{M, +} \Sigma^{-1}_{s, \mc{T}, +}  A_{M, +} B^{r, T}_{\mc{T}}  \left(B^{r, -1}_{\mc{T}} \mb{f} - B^{r, -1}_{\mc{T}} \mb{f}(\mb{\hat{x}}, \mc{T}) \right)  \nonumber \\
	    & \hspace{30mm} - \frac{1}{2} \ln \left( \det( \Sigma^{-1}_{s, \mc{T}, +} )\right)           \label{s-to-f-eq5}   \\	    	    	    	    	    	    
	    &= \underset{\mc{T} \in \mbb{T}^{+},~ \mb{f}: A_{M, +} \mb{f} = \mb{s}_{obs, +} }{\arg\min} \frac{1}{2} \left(  B^{r}_{\mc{T}} \mb{f} - \mb{\hat{x}}  \right)^{T} B^{r, T, -1}_{\mc{T}} A^T_{M, +} \Sigma^{-1}_{s, \mc{T}, +}  A_{M, +} B^{r, -1}_{\mc{T}} \left(B^{r}_{\mc{T}} \mb{f} - \mb{\hat{x}} \right) 
	         - \frac{1}{2} \ln \left( \det( \Sigma^{-1}_{s, \mc{T}, +} )\right).  \label{s-to-f-eq6}     	\\    	    	    
	     \hline \nonumber
\end{align}
\end{figure*}  

The following reductions are performed.
Eq. \eqref{s-to-f-eq1} restates the combinatorial detector in terms of the non-zero observations and invertible covariance matrix.
Eq. \eqref{s-to-f-eq2} replaces the observed flow $\mb{s}_{obs, +}$, with an unknown flow to be determined $\mb{f}$, under the constraint that $A_{\mc{M}, +} \mb{f} = \mb{s}_{obs, +}$.
The second term in the quadratic form, $\mb{s}_{+}(\mb{\hat{x}}, \mc{T})$ which is the predicted observation under a hypothesis, is replaced with $A_{\mc{M}, +} \mb{f}(\mb{x}, \mc{T}^{\star})$.
Eq. \eqref{s-to-f-eq3} - \eqref{s-to-f-eq3} rearrange terms and push a $B^{r, T, -1}_{\mc{T}}$ into the quadratic form.

We can re-write \eqref{s-to-f-eq6} in the form

\begin{align}
& \text{min}~ \frac{1}{2}(  B^{r}_{\mc{T}} \mb{f} - \mb{\hat{x}} ) \Sigma^{1}_{\mc{T}}(  B^{r}_{\mc{T}} \mb{f} - \mb{\hat{x}} )  - \frac{1}{2} \ln \left( \det( \Sigma^{2}_{\mc{T} } )\right) 	\label{apndx-opt3} \tag{OPT-3}  	\\
& \text{s.t.} 																																\nonumber 				\\
& ~~~~~ A_{M, +} \mb{f} = \mb{s}_{obs,+}      																										\label{apndx-opt3-flow-obs}  		\\
& ~~~~~ \mc{T} \in \mbb{T}_{+}.        																												\label{apndx-opt3-sp-tree}
\end{align}

We can consider this form to be a stochastic equivalent to the deterministic MILP in \eqref{opt1}, with the following matrices:
\begin{align*}
\Sigma^{2}_{ \mc{T} } &= \sigma^2 A_{\mc{M},+} B^{r, -1}_{\mc{T}} B^{r,-1, T}_{\mc{T}} A^T_{\mc{M},+}, \\
\Sigma^{1}_{ \mc{T} } &=  B^{r, T, -1}_{\mc{T}} A^T_{\mc{M}, +}  \left( \Sigma^{2}_{ \mc{T} }  \right)^{-1} A_{\mc{M}, +} B^{r, T}_{\mc{T}}. \\
\end{align*}

\subsection{Proof of Theorem \ref{thm:MST-OPT-BOUNDS} }
\label{subsection-MST-upper-bound}

Starting from \eqref{OPT_F_L4}, we have the following:

\begin{align}
\text{OPT}(\mb{f}_{obs}) &= \underset{ \mc{T} \in \mbb{T}^{+} } {\min} \frac{1}{2} \|  B_{G \setminus \mc{T} } \mb{f}_{obs}  \|^{2}   						\label{THRM_4_APP_L1}  \\
				     &= \underset{ \mc{T} \in \mbb{T}^{+} } {\min} \frac{1}{2} \| \sum_{i: e_i \in G \setminus \mc{T}} b_{i} \mb{f}_{obs}(i) \|^2  				\label{THRM_4_APP_L2}  \\ 
				     &\leq \underset{ \mc{T} \in \mbb{T}^{+} } {\min} \frac{1}{2} \sum_{i: e_i \in G \setminus \mc{T}}  \|  b_{i} \mb{f}_{obs}(i) \|^2  			\label{THRM_4_APP_L3}  \\
				     &= \underset{ \mc{T} \in \mbb{T}^{+} } {\min} \sum_{i: e_i \in G \setminus \mc{T}} \| \mb{f}_{obs}(i) \|^2.							\label{THRM_4_APP_L4}  \\
				     &= \sum_{i: e_i \in G } \| \mb{f}_{obs}(i) \|^2 -  \underset{ \mc{T} \in \mbb{T}^{+} } {\min} \sum_{i: e_i \in \mc{T}} \| \mb{f}_{obs}(i) \|^2	\label{THRM_4_APP_L5}  \\
				     &= \|\mb{f}_{obs}\|^2 - \text{MST}(- |\mb{f}_{obs}|^2).																\label{THRM_4_APP_L6} 				     
\end{align}
  
Eq. \eqref{THRM_4_APP_L2} represents the partitioned incidence matrix as a sum of column vectors. 
Inequality in \eqref{THRM_4_APP_L3}, arises from the triangle inequality.
Since, each column vector of the incidence matrix is a $+1$, $-1$ pair, this reduces to \eqref{THRM_4_APP_L4}.
Minimizing the sum of squares of each co-tree weights in \eqref{THRM_4_APP_L4} is equivalent to MST$(- |\mb{f}_{obs}|^2)$, which is equivalent to MST$(- |\mb{f}_{obs}|)$ since the greedy edge addition step in solving a minimum spanning tree problem will take the same action regardless if the edge weights are squared or not.

%
%
\subsection{Hypothesis Testing Interpretation of Flow Based Approximate ML}
\label{subsection-alternative-view-of-FMST}

In Section \ref{subsection-flow-based-APX-MAP-detector}, the approximate ML detector was formulated as an MIQP.
An alternative interpretation of this optimization is that of a hypothesis test of the noisy flows being actually of value zero.

The flow solution in  \eqref{network-flow-constraint-matrix} can be used to construct an efficient hypothesis detector which has polynomial run-time.
In the stochastic case $\mb{x}$ is not known, but $\mb{\hat{x}}$ is given, therefore $\mb{f}(\mb{\hat{x}}, \mb{s} )$ can be used.

First, we can determine the distribution of this \textit{noisy-flow} vector conditioning on a candidate hypothesis $\mc{T}_i$:  
\begin{align} 
\mb{f}(\hat{\mb{x}}, \mb{s}_{obs} ) &= B^{r, -1}_N (\hat{\mb{x}} - B^{r}_{M} \mb{s}_{obs} )  										\label{eq:noisy-flow-line1} \\ 
					  	    &= B^{r, -1}_N (\hat{\mb{x}} - B^{r}_{M} ( \mb{s}(\mc{T}_i, \mb{\hat{x}}) + \epsilon_{s, i}) ) 			\label{eq:noisy-flow-line2} \\
						    &= \mb{f}(\mb{\hat{x}},  \mb{s}(\mc{T}_i, \mb{\hat{x}}) ) + B^{r, -1}_{N} B^{r}_{M} \epsilon_{s, i} 	  	 \label{eq:noisy-flow-line3} \\
						    & \sim N ( \mb{f}(\mb{\hat{x}}, \mb{s}(\mc{T}_i, \mb{\hat{x}}) ),  \Sigma_{f, i} ): \mc{T}_i \text{ is true}   	  \label{eq:noisy-flow-line4}
\end{align}

The LHS of  \eqref{eq:noisy-flow-line1} is the distribution of the \textit{noisy-flow} conditioning on a particular hypothesis $\mc{T}_i$.
The RHS evaluates the flow network solution \eqref{f-star-solution} using the forecasted consumption $\mb{\hat{x}}$ instead of the true value $\mb{x}$.
This relies on  \eqref{eq:stochastic-flow-obs-eq2}-\eqref{eq:stochastic-flow-obs-eq4}.

The vector $\mb{f}(\mb{\hat{x}}, \mb{s}(\mc{T}_i, \mb{\hat{x}}) )$ is the flow from spanning tree $\mc{T}_i$ and nodal injections $\mb{\hat{x}}$.
The true `noisy-flow' is distributed around this value.
   
A possible hypothesis test is the following: 
\begin{align}   
\hat{\mc{T}} = \underset{\forall \mc{T}_i \in \mbb{T} }{\arg\max}~\Pr( \mb{f}(\hat{\mb{x}}, \mb{s}_{obs} )~|~\mb{f}(\mb{\hat{x}}, \mb{s}(\mc{T}_i, \mb{\hat{x}}) ) ) \label{eq-zero-flow-test-full}.
\end{align}

This is no better than \eqref{eq-combinatorial-MAP-detector} for the following reasons:
\begin{enumerate}
\item We must compute the hypothesis mean $\mb{f}(\mb{\hat{x}}, \mb{s}(\mc{T}_i, \mb{\hat{x}}) )$ under every spanning tree $\mc{T}_i$, it is still of $O(|\mbb{T}|)$ complexity.
\item The covariance matrix is of rank $\mu$ (the rank of $\Gamma_i$) and not $E$ (the size of $\mb{f}$).    
         Therefore $\Sigma^{-1}_{f,i}$ is not positive definite and therefore the inverse $\Sigma^{-1}_{f, i}$ cannot be computed. 
\end{enumerate}   
   
We can alternatively test $\mu$ elements of the \textit{noisy-flow} $\mb{f}(\hat{\mb{x}}, \mb{s} )$ under the hypothesis that their true value is zero corresponding to the zero's of the hypothesized spanning tree.

Using the following shorthand:  $\mb{f}_{obs} = \mb{f}(\hat{\mb{x}}, \mb{s}_{obs} )$ and $\mb{f}_{\mc{T}} = \mb{f}(\hat{\mb{x}}, \mb{s}(\mc{T}, \mb{\hat{x}}) )$ we can represent the variables in \eqref{eq-zero-flow-test-full}.
Consider the set of indices $i \in \mc{I}_{\mc{T}}$ where $\mc{I}_{\mc{T}} = \{i: e_i \in E \setminus \mc{T} \}$.
This can be used to index into the vector $\mb{f}_{obs}$ and $\mb{f}_{\mc{T}}$.
Therefore, under a hypothesis $\mc{T}$, we can calculate the following likelihood : 
\begin{align}
\hat{\mc{T}} =  \underset{\mc{I}_{\mc{T}}: \mc{T} \in \mbb{T} }{\arg\max} \Pr( \mb{f}_{obs}(i_1), \hdots, \mb{f}_{obs}(i_\mu) | \mb{f}_{\mc{T}}(i_1) = 0, \hdots, \mb{f}_{\mc{T}}(i_\mu) = 0 ).  \label{eq-zero-flow-detector}
\end{align}

In this case, the reduced covariance matrix is potentially invertible.
The following theorem relates the combinatorial test to a test of zero flows on the empirical flow.
\begin{thm}
\label{thm-zero-flow}
The zero flow hypothesis detection in \eqref{eq-zero-flow-detector} and the combinatorial flow hypothesis test in  \eqref{eq-combinatorial-MAP-detector} are equivalent for $\mc{M}$ satisfying the placement condition in Theorem \ref{thm:spanning_tree_identifiability} of size $\mu$.
\end{thm}

\begin{proof}
We construct this test statistic 
\begin{align}
B_H ( \mb{f}_{obs} - \mb{f}_{\mc{T}} ) = [ \mb{f}_{obs}(i_{1}) , \hdots,  \mb{f}_{o}(i_{\mu})  ]^{T},
\end{align}
with matrix $B_H \in \{0, 1 \}^{|\mu| \times |E|}$ where $B_H(k, i_k) = 1$ for each edge in the co-tree of the particular spanning tree under hypothesis.

The test statistic can be reduced to the following:
\begin{align}
B_H ( \mb{f}_{obs} - \mb{f}_{\mc{T}} ) &= B_H \left( (B^{r,-1}_N) \hat{\mb{x}}  - (B^{r,-1}_N) B^r_M \mb{s}  \right)  \\ 
	       					      	  &- B_H \left( (B^{r,-1}_N) \hat{\mb{x}}  - (B^{r,-1}_N) B^r_M \mb{s}(\mc{T}, \mb{\hat{x}})  \right) \\
	       						  &= \underbrace{ B_H  (B^{r,-1}_N) B^r_M }_{H} \left( \mb{s}  - \mb{s}(\mc{T}, \mb{\hat{x}} ) \right) 
\end{align}

Therefore, the new hypothesis test using $x(\mc{T})$, is equivalent to the previous test using observed flow with some linear transformation $H$.
The matrix is $H$ is full rank in the $|\mc{M}| = \mu$ case. 
Clearly if H is not square when $|\mc{M}| > \mu$, $H$ will not be full rank and the two tests will no longer be the same.
\end{proof}

Therefore the ML detector over the set of all spanning trees reduces to evaluating the probability that a subset of these vectors are should actually be zero.
Unfortunately we still need to enumerate $|\mbb{T}|$ hypotheses.
Fortunately, however, this flow vector form allows us to very efficiently prune out all but a few alternative hypothesis to test.
Intuitively, $\mb{f}(\hat{\mb{x}}, \mb{s}_{obs} )$ will have some very few small values which actually encode potential spanning trees, and many very large values which can just be pruned.

Therefore, the approximate solution based on noisy flow minimum spanning tree detector is approximating the hypothesis testing procedure of determining which edges have zero underlying flow.

\bibliographystyle{IEEEtran}
\bibliography{tps_bib}

\begin{thebibliography}{10}
\providecommand{\url}[1]{#1}
\csname url@rmstyle\endcsname
\providecommand{\newblock}{\relax}
\providecommand{\bibinfo}[2]{#2}
\providecommand\BIBentrySTDinterwordspacing{\spaceskip=0pt\relax}
\providecommand\BIBentryALTinterwordstretchfactor{4}
\providecommand\BIBentryALTinterwordspacing{\spaceskip=\fontdimen2\font plus
\BIBentryALTinterwordstretchfactor\fontdimen3\font minus
  \fontdimen4\font\relax}
\providecommand\BIBforeignlanguage[2]{{%
\expandafter\ifx\csname l@#1\endcsname\relax
\typeout{** WARNING: IEEEtran.bst: No hyphenation pattern has been}%
\typeout{** loaded for the language `#1'. Using the pattern for}%
\typeout{** the default language instead.}%
\else
\language=\csname l@#1\endcsname
\fi
#2}}

\bibitem{Lavaei2012}
J.~Lavaei, D.~Tse, and B.~Zhang, ``Geometry of power flows in tree networks,''
  in \emph{Power and Energy Society General Meeting, 2012 IEEE}.\hskip 1em plus
  0.5em minus 0.4em\relax IEEE, 2012, pp. 1--8.

\bibitem{Lam2011}
A.~Lam, B.~Zhang, and D.~Tse, ``Distributed algorithms for optimal power flow
  problem,'' \emph{arXiv preprint arXiv:1109.5229}, 2011.

\bibitem{Farivar2012}
M.~Farivar, R.~Neal, C.~Clarke, and S.~Low, ``Optimal inverter var control in
  distribution systems with high pv penetration,'' in \emph{Power and Energy
  Society General Meeting, 2012 IEEE}.\hskip 1em plus 0.5em minus 0.4em\relax
  IEEE, 2012, pp. 1--7.

\bibitem{Lam2012}
A.~Lam, A.~Dominguez-Garcia, B.~Zhang, and D.~Tse, ``Optimal distributed
  voltage regulation in power distribution networks,'' Tech. Rep., 2012.

\bibitem{Jahangiri2013}
P.~Jahangiri and D.~C. Aliprantis, ``Distributed volt/var control by pv
  inverters,'' \emph{Power Systems, IEEE Transactions on}, vol.~28, no.~3, pp.
  3429--3439, 2013.

\bibitem{Smith2011}
J.~Smith, W.~Sunderman, R.~Dugan, and B.~Seal, ``Smart inverter volt/var
  control functions for high penetration of pv on distribution systems,'' in
  \emph{Power Systems Conference and Exposition (PSCE), 2011 IEEE/PES}.\hskip
  1em plus 0.5em minus 0.4em\relax IEEE, 2011, pp. 1--6.

\bibitem{Monticelli2000}
A.~Monticelli, ``Electric power system state estimation,'' \emph{Proceedings of
  the IEEE}, vol.~88, no.~2, pp. 262--282, 2000.

\bibitem{Korres2012}
G.~N. Korres and N.~M. Manousakis, ``A state estimation algorithm for
  monitoring topology changes in distribution systems,'' in \emph{Power and
  Energy Society General Meeting, 2012 IEEE}.\hskip 1em plus 0.5em minus
  0.4em\relax IEEE, 2012, pp. 1--8.

\bibitem{Arghandeh2015}
R.~Arghandeh, M.~Gahr, A.~von Meier, G.~Cavraro, M.~Ruh, and G.~Andersson,
  ``Topology detection in microgrids with micro-synchrophasors,'' \emph{arXiv
  preprint arXiv:1502.06938}, 2015.

\bibitem{Cavraro2015}
G.~Cavraro, R.~Arghandeh, G.~Barchi, and A.~von Meier, ``Distribution network
  topology detection with time-series measurements,'' in \emph{Innovative Smart
  Grid Technologies Conference (ISGT), 2015 IEEE Power \& Energy
  Society}.\hskip 1em plus 0.5em minus 0.4em\relax IEEE, 2015, pp. 1--5.

\bibitem{Deka2015A}
D.~Deka, S.~Backhaus, and M.~Chertkov, ``Structure learning in power
  distribution networks: Part i,'' \emph{arXiv preprint arXiv:1502.07820},
  2015.

\bibitem{Deka2015B}
------, ``Structure learning in power distribution networks: Part ii,''
  \emph{arXiv preprint arXiv:1502.07820}, 2015.

\bibitem{Sharon2012}
Y.~Sharon, A.~M. Annaswamy, A.~L. Motto, and A.~Chakraborty, ``Topology
  identification in distribution network with limited measurements,'' in
  \emph{Innovative Smart Grid Technologies (ISGT), 2012 IEEE PES}.\hskip 1em
  plus 0.5em minus 0.4em\relax IEEE, 2012, pp. 1--6.

\bibitem{Stott2009dc}
B.~Stott, J.~Jardim, and O.~Alsa{\c{c}}, ``Dc power flow revisited,''
  \emph{IEEE Transactions on Power Systems}, vol.~24, no.~3, pp. 1290--1300,
  2009.

\bibitem{Williamson2011}
D.~P. Williamson and D.~B. Shmoys, \emph{The design of approximation
  algorithms}.\hskip 1em plus 0.5em minus 0.4em\relax Cambridge university
  press, 2011.

\bibitem{Gabow1978}
H.~N. Gabow and E.~W. Myers, ``Finding all spanning trees of directed and
  undirected graphs,'' \emph{SIAM Journal on Computing}, vol.~7, no.~3, pp.
  280--287, 1978.

\bibitem{Diestel2000}
R.~Diestel, \emph{Graph Theory: Graduate Texts in Mathematics}.\hskip 1em plus
  0.5em minus 0.4em\relax Springer-Verlag Berlin and Heidelberg GmbH \& Company
  KG, 2000.

\bibitem{nemhauser1988}
G.~L. Nemhauser and L.~A. Wolsey, \emph{Integer and combinatorial
  optimization}.\hskip 1em plus 0.5em minus 0.4em\relax Wiley New York, 1988,
  vol.~18.

\bibitem{Sevlian2014}
R.~Sevlian, S.~Patel, and R.~Rajagopal, ``Distribution system load and forecast
  model,'' \emph{arXiv preprint arXiv:1407.3322}, 2014.

\bibitem{Syslo1982}
M.~M. Syslo, ``On the fundamental cycle set graph,'' \emph{Circuits and
  Systems, IEEE Transactions on}, vol.~29, no.~3, pp. 136--138, 1982.

\bibitem{GraphsAndMatrices2010}
R.~B. Bapat, \emph{Graphs and matrices}.\hskip 1em plus 0.5em minus 0.4em\relax
  Springer, 2010.

\end{thebibliography}

\end{document}